\documentclass[runningheadsm, envcountsect]{llncs}
\usepackage{amsmath}
\usepackage{amssymb}
\usepackage{todonotes}
\usepackage{stmaryrd}
\usepackage{xspace}
\usepackage{listings}
\usepackage{comment}
\usepackage{mathtools}
\usepackage{environ}
\usepackage{enumitem}
\usepackage{tikz}
\usetikzlibrary{trees}
\usepackage{multirow}
\usepackage{array}
\usepackage{graphics}
\usepackage[cmtip,all]{xy}
\usepackage{units}
\usepackage{placeins}
\usepackage{lineno}
\usepackage{cite}
\usepackage{wrapfig}

\newcommand{\longsquiggly}{\xymatrix@C-=0.5cm{{}\ar@{~>}[r]&{}}}

\newcommand{\sectionword}{Section}
\newcommand{\appendixword}{Appendix}
\newcommand{\figureword}{Fig.}
\newcommand{\lemmaword}{Lem.}
\newcommand{\eqdef}{\overset{\mathrm{def}}{=}}
\newcommand{\ruleno}[1]{\mbox{[\textsc{#1}]}}
\newcommand{\bigslant}[2]{{\raisebox{.2em}{$#1$}\left/\raisebox{-.2em}{$#2$}\right.}}

\newcommand{\argmax}[1]{\underset{\makebox[0pt]{\mbox{\normalfont\tiny {$#1$}}}}{argmax} \;\;}
\newcommand{\maxd}{\max^{\bullet}}
\newcommand{\ttt}[1]{\mbox{\texttt{#1}}}
\let\emptyset\varnothing

\newcommand{\inbr}[1]{\left<{#1}\right>}
\newcommand{\fv}[1]{\mathcal{FV}\,({#1})}

\newcommand{\grterms}{\mathcal{T}_{\emptyset}}
\newcommand{\Dom}{\mathcal{D}om}
\newcommand{\VRan}{\mathcal{VR}an}

\newcommand{\substitute}[3]{#1\,[\bigslant{#2}{#3}]}

\newcommand{\taskst}[2]{\langle #1 ,\, #2 \rangle}
\newcommand{\mkenv}[2]{(#1 ,\, #2)}
\newcommand{\unigoal}[2]{#1 \equiv #2}
\newcommand{\conjgoal}[2]{#1 \land #2}
\newcommand{\disjgoal}[2]{#1 \lor #2}
\newcommand{\freshgoal}[2]{\mbox{\lstinline|fresh|} \, #1\;.\; #2}
\newcommand{\invokegoal}[3]{#1\,(#2, \, \dots, \, #3)}

\newcommand{\ninit}{n_{init}}

\newcommand{\schemetrans}[6]{\withenv{#2,\,#3,\,#4,\,#5}{#1} \longsquiggly #6}
\newcommand{\schemewithvset}[2]{{#1}^{#2}}

\newcommand{\schemenode}[1]{
\tikz[anchor=base,baseline,sibling distance=3cm,edge from parent/.style={draw,-latex}]  \node{#1};
 }

\newcommand{\schemesarrow}[3]{
\begin{tikzpicture}[level distance=30pt,sibling distance=3cm,edge from parent/.style={draw,-latex}]
  \node{#1}child {node{#3}edge from parent node[right]{\tiny{#2}}};
\end{tikzpicture}}

\newcommand{\schemedarrow}[3]{
\begin{tikzpicture}[level distance=30pt,sibling distance=3cm,edge from parent/.style={draw,-latex}]
  \node{#1}child {node{#3}edge from parent node[right]{\tiny{#2}}};
\end{tikzpicture}}

\newcommand{\schemefork}[2]{
  \begin{tikzpicture}[level distance=20pt,edge from parent/.style={draw,-latex}]
   \coordinate   
      child {node{#1}}
      child {node{#2}} ;
\end{tikzpicture}}

\newcommand{\upd}[2]{\mbox{\textbf{upd}}\,(#1,\, #2)}
\newcommand{\constr}[2]{\mbox{\textbf{constr}}\,(#1,\, #2)}

\newcommand{\withenv}[2]{\left< #1 \right> \;\vdash\; #2}
\newcommand{\onepremrule}[2]{\dfrac{#1}{#2}}
\newcommand{\twopremrule}[3]{\dfrac{#1,\; #2}{#3}}
\newcommand{\threepremrule}[4]{\dfrac{#1,\;  #2,\;  #3}{#4}}

\newcommand{\lazystream}[1]{\texttt{Lazy [{#1}]}}
\newcommand{\consstream}[2]{\texttt{Cons #1 [{#2}]}}

\newcommand{\sembr}[1]{\llbracket #1 \rrbracket}
\newcommand{\tra}[1]{\mathcal{T}r^{ans}(#1)}
\newcommand{\trs}[1]{\mathcal{T}r^{st}(#1)}

\newcommand{\mK}{\textsc{miniKanren}\xspace}

\newcommand{\costdisj}[2]{cost_{\oplus}(#1 \oplus #2)}
\newcommand{\costconj}[2]{cost_{\otimes}(#1 \otimes #2)}
\newcommand{\lookuptime}[1]{\texttt{lookup}\,(#1)}
\newcommand{\addtime}[1]{\texttt{add}\,(#1)}
\renewcommand{\O}{\mathcal{O}}

\renewcommand{\labelenumii}{\arabic{enumi}.\arabic{enumii}}
\renewcommand{\labelenumiii}{\arabic{enumi}.\arabic{enumii}.\arabic{enumiii}}
\renewcommand{\labelenumiv}{\arabic{enumi}.\arabic{enumii}.\arabic{enumiii}.\arabic{enumiv}}

\let\lemma\relax
\spnewtheorem{lemma}{Lemma}[section]{\bfseries}{\itshape}

\newcommand{\repeatlemma}[1]{%
  \begingroup
  \renewcommand{\thelemma}{\ref{#1}}%
  \expandafter\expandafter\expandafter\lemma
  \csname replemma@#1\endcsname
  \endlemma
  \endgroup
}

\newcommand{\repeattheorem}[1]{%
  \begingroup
  \renewcommand{\thetheorem}{\ref{#1}}%
  \expandafter\expandafter\expandafter\theorem
  \csname reptheorem@#1\endcsname
  \endtheorem
  \endgroup
}

\NewEnviron{replemma}[1]{%
  \global\expandafter\xdef\csname replemma@#1\endcsname{%
    \unexpanded\expandafter{\BODY}%
  }%
  \expandafter\lemma\BODY\label{#1}\endlemma
}

\NewEnviron{reptheorem}[1]{%
  \global\expandafter\xdef\csname reptheorem@#1\endcsname{%
    \unexpanded\expandafter{\BODY}%
  }%
  \expandafter\theorem\BODY\label{#1}\endtheorem
}

\lstdefinelanguage{minikanren}{
keywords={fresh},
sensitive=true,
commentstyle=\small\itshape\ttfamily,
keywordstyle=\textbf,
identifierstyle=\ttfamily,
basewidth={0.5em,0.5em},
columns=fixed,
fontadjust=true,
literate={fun}{{$\lambda\;\;$}}1 {->}{{$\to$}}3 {===}{{$\,\equiv\,$}}1 {=/=}{{$\not\equiv$}}1 {|>}{{$\triangleright$}}3 {/\\}{{$\wedge$}}2 {\\/}{{$\vee$}}2,
morecomment=[s]{(*}{*)}
}

\usepackage{letltxmacro}
\newcommand*{\SavedLstInline}{}
\LetLtxMacro\SavedLstInline\lstinline
\DeclareRobustCommand*{\lstinline}{%
  \ifmmode
    \let\SavedBGroup\bgroup
    \def\bgroup{%
      \let\bgroup\SavedBGroup
      \hbox\bgroup
    }%
  \fi
  \SavedLstInline
}

\begin{document}


\makeatletter
\let\origsection\section
\renewcommand\section{\@ifstar{\starsection}{\nostarsection}}

\newcommand\nostarsection[1]
{\sectionprelude\origsection{#1}\sectionpostlude}

\newcommand\starsection[1]
{\sectionprelude\origsection*{#1}\sectionpostlude}

\newcommand\sectionprelude{%
  \vspace{-2mm}
}

\newcommand\sectionpostlude{%
  \vspace{-2mm}
}
\makeatother

\setlength{\abovecaptionskip}{-5pt plus 3pt minus 2pt}
\setlength{\belowcaptionskip}{-20pt plus 3pt minus 2pt}

\abovedisplayskip-1mm
\belowdisplayskip0mm
\abovedisplayshortskip-3mm
\belowdisplayshortskip0mm

\setlength{\topsep}{0pt}
\setlength{\partopsep}{0pt plus 0pt minus 0pt}
\setlength{\parskip}{0pt}

\title{Scheduling Complexity of Interleaving Search}

\author{Dmitry Rozplokhas\orcidID{0000-0001-7882-4497} \and \\
Dmitry Boulytchev\orcidID{0000-0001-8363-7143}}
\authorrunning{D. Rozplokhas and D. Boulytchev}
%
\institute{St Petersburg University and JetBrains Research, Russia \\ \email{rozplokhas@gmail.com}, \email{dboulytchev@math.spbu.ru}}

\maketitle

\begin{abstract}
  \mK is a lightweight embedded language for logic and relational programming. Many of its useful features come from
  a distinctive search strategy, called \emph{interleaving search}. However, with interleaving search conventional  
  ways of reasoning about the complexity and performance of logical programs become irrelevant. We identify an important
  key component~--- \emph{scheduling}~--- which makes the reasoning for \mK so different, and present a semi-automatic
  technique to estimate the scheduling impact via symbolic execution for a reasonably wide class of programs.
\end{abstract}

\section{Introduction}
\label{sec:intro}

\begin{figure}[t]
\begin{tabular}{p{6cm}p{6cm}}
\begin{lstlisting}[basicstyle=\small]
   append$^o$ = fun a b ab .
     ((a === Nil) /\ (ab === b)) \/
     (fresh (h t tb)
        (a === Cons(h, t)) /\
        (append$^o$ t b tb) /\
        (ab === Cons(h, tb)))
\end{lstlisting} & \multirow{2}{*}[-3mm]{\includegraphics[width=6cm,height=5cm]{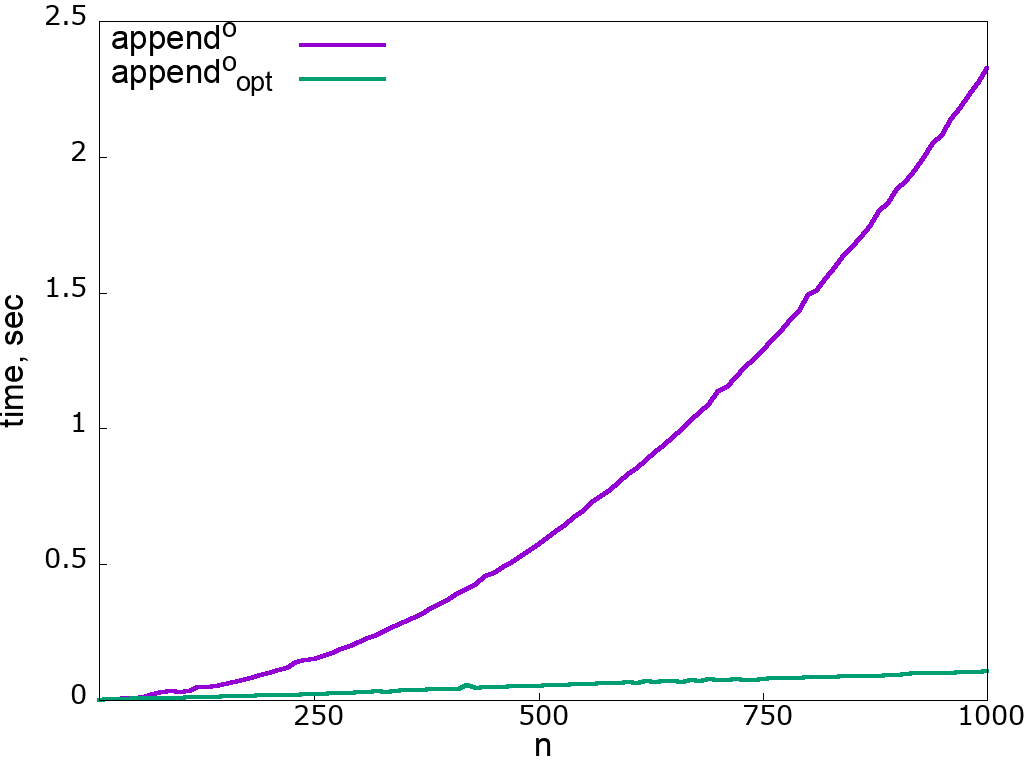}} \\[-9mm]
\begin{lstlisting}[basicstyle=\small]
   append$^o_{opt}$ = fun a b ab .
     ((a === Nil) /\ (ab === b)) \/
     (fresh (h t tb)
        (a === Cons(h, t)) /\
        (ab === Cons(h, tb) /\
        (append$^o_{opt}$ t b tb)))
\end{lstlisting} &
\end{tabular}
\caption{Two implementations of list concatenation and their performance for $a = [1,\dots,n]$, $b = [1,\dots,100]$, and $ab$ left free.}
\label{fig:length_implementations}
\end{figure}

A family of embedded languages for logic and, more specifically, relational programming \mK~\cite{TRS} has demonstrated an interesting potential in various fields of 
program synthesis and declarative programming~\cite{SevenProblems,Quines,Matching}. A distinctive feature of \mK is \emph{interleaving search}~\cite{Transformers} which,
in particular, delivers such an important feature as completeness.

However, being a different search strategy than conventional BFS/DFS/iterative deepening, etc., interleaving search makes the conventional ways of reasoning about the complexity
of logical programs not applicable. Moreover, some intrinsic properties of interleaving search can manifest themselves in a number of astounding and, at the first glance, unexplainable
performance effects. 

As an example, let's consider two implementations of list concatenation relation (\figureword~\ref{fig:length_implementations}, left side); we respect here a conventional tradition
for \mK programming to superscript all relational names with ``o''. The only difference between the two is
the position of the recursive call.
The evaluation of these implementations on the same problem (\figureword~\ref{fig:length_implementations}, right side) shows that the first implementation works significantly
slower,
although it performs exactly the same number of unifications. As a matter of fact, these two implementations even have different \emph{asymptotic} complexity under the assumption
that occurs check is disabled.\footnote{The role of occurs check is discussed in
\sectionword~\ref{sec:evaluation}.} Although the better performance of the \lstinline|append$_{opt}^o$| relation is expected even under conventional strategies due to tail recursion, the asymptotic difference is striking.

A careful analysis discovers that the difference is caused not by unifications, but by the process of \emph{scheduling} goals during the search. In \mK a
lazy structure is maintained to decompose the goals into unifications, perform these unifications in a certain order, and thread the results appropriately. For the \lstinline|append$_{opt}^o$| relation the size of this structure is constant, while for the \lstinline|append$^o$|
this structure becomes linear in size, reducing the performance.


This paper presents a formal framework for scheduling cost complexity analysis for interleaving search in \mK.
We use the reference operational semantics, reflecting the behaviour of actual implementations~\cite{CertifiedSemantics}, and prove the soundness of our approach w.r.t. this semantics.
The roadmap of the approach is as follows: we identify two complexity measures (one of which captures \emph{scheduling complexity}) and give exact and approximate recursive formulae to calculate them (\sectionword~\ref{sec:scheduling}); then we present a procedure to automatically extract inequalities for the measures for a given goal using symbolic execution (\sectionword~\ref{sec:symbolic}). 
These inequalities have to be reformulated and solved manually in terms of a certain \emph{metatheory}, which, on success, provides asymptotic bounds for the scheduling complexity of a goal evaluation.
Our approach puts a number of restrictions on the goal being analyzed as well as on the relational program as a whole.
We explicitly state these restrictions in \sectionword~\ref{sec:background} and discuss their impact in \sectionword~\ref{sec:discussion}.
The proofs of all lemmas and theorems can be found in \appendixword~\ref{sec:proofs_appendix}.

\section{Background: Syntax and Semantics of \mK}
\label{sec:background}

In this section, we recollect some known formal descriptions for \mK language that will be used as a basis for our development.
The descriptions here are taken from~\cite{CertifiedSemantics} (with a few non-essential adjustments for presentation purposes) to make
the paper self-contained, more details and explanations can be found there.

The syntax of core \mK is shown in Fig.~\ref{fig:syntax}. 
All data is presented using terms $\mathcal{T}_X$ built from a fixed set of constructors $\mathcal{C}$ with known arities and variables
from a given set $X$.
We parameterize the terms with an alphabet of variables since in the semantic description we will need \emph{two} kinds of variables:
\emph{syntactic} variables $\mathcal{X}$, used for bindings in the definitions, and \emph{logic} variables $\mathcal{A}$, which are
introduced and unified during the evaluation. We assume the set $\mathcal{A}$ is ordered and use the notation $\alpha_i$ 
to specify a position of a logical variable w.r.t. this order.

There are five types of goals: unification of two terms, conjunction and disjunction of goals,
fresh logic variable introduction, and invocation of some relational definition. For the sake of brevity, in code snippets, we abbreviate
immediately nested ``\lstinline|fresh|'' constructs into the one, writing ``\lstinline|fresh x y $\dots$ . $g$|'' instead of
``\lstinline|fresh x . fresh y . $\dots$ $g$|''. The \emph{specification} $\mathcal{S}$ consists of a set of relational definitions and a top-level goal.
A top-level goal represents a search procedure that returns a stream of substitutions for the free variables of the goal.

\begin{figure}[t]
\centering
\[
\begin{array}{ccll}
  \mathcal{C} & = & \{C_i^{k_i}\} & \mbox{constructors} \\
  \mathcal{T}_X & = & X \cup \{C_i^{k_i} (t_1, \dots, t_{k_i}) \mid t_j\in\mathcal{T}_X\} & \mbox{terms over the set of variables $X$} \\
  \mathcal{D} & = & \mathcal{T}_\emptyset & \mbox{ground terms}\\
  \mathcal{X} & = & \{ \ttt{x}, \ttt{y}, \ttt{z}, \dots \} & \mbox{syntactic variables} \\
  \mathcal{A} & = & \{ x, y, z \dots \} & \mbox{logic variables} \\
  \mathcal{R} & = & \{ R_i^{k_i}\} &\mbox{relational symbols with arities} \\[2mm]
  \mathcal{G} & = & \mathcal{T_X}\equiv\mathcal{T_X}   &  \mbox{equality} \\
              &   & \mathcal{G}\wedge\mathcal{G}     & \mbox{conjunction} \\
              &   & \mathcal{G}\vee\mathcal{G}       &\mbox{disjunction} \\
              &   & \mbox{\lstinline|fresh|}\;\mathcal{X}\;.\;\mathcal{G} & \mbox{fresh variable introduction} \\
              &   & R_i^{k_i} (t_1,\dots,t_{k_i}),\;t_j\in\mathcal{T_X} & \mbox{relational symbol invocation} \\[2mm]
  \mathcal{S} & = & \{R_i^{k_i} = \lambda\;\ttt{x}_1^i\dots \ttt{x}_{k_i}^i\,.\, g_i;\}\; g, \; g_i, g\in\mathcal{G}\phantom{XXX} & \mbox{specification}
\end{array}
\]
\caption{The syntax of \mK}
\label{fig:syntax}
\end{figure}

During the evaluation of \mK program an environment, consisting of a substitution for logic variables and a counter of allocated logic
variables, is threaded through the computation and updated in every unification and fresh variable introduction.
The substitution in the environment at a given point and given branch of evaluation contains all the information about relations between
the logical variables at this point.
Different branches are combined via \emph{interleaving search} procedure~\cite{Transformers}.
The answers for a given goal are extracted from the final environments.

\begin{figure}[t]
\centering
\[
\begin{array}{ccllcccll}
  \Sigma & = & \mathcal{A} \to \mathcal{T}_\mathcal{A} & \mbox{substitutions} &\qquad\qquad&        S & = & \taskst{\mathcal{G}}{E} & \mbox{task} \\
       E & = & \Sigma \times \mathbb{N}              & \mbox{environments}  &\qquad\qquad&          &   & S \oplus S              & \mbox{sum} \\
         &   &                                       &                      &\qquad\qquad&          &   & S \otimes \mathcal{G}   & \mbox{product} \\
       L & = & \circ \; \mid \; E      & \mbox{labels}        &\qquad\qquad&  \hat{S} & = & \diamond \; \mid \; S   & \mbox{states} \\
\end{array}
\]
\caption{States and labels in the LTS for \mK}
\label{fig:operanional_semantics_states_labels}
\end{figure}

This search procedure is formally described by operational semantics in the form of a labeled transition system.
This semantics corresponds to the canonical implementation of interleaving search. 

The form of states and labels in the transition system is defined in \figureword~\ref{fig:operanional_semantics_states_labels}.
Non-terminal states $S$ have a tree-like structure with intermediate nodes corresponding to partially evaluated conjunctions
(``$\otimes$'') or disjunctions (``$\oplus$'').
A leaf in the form $\taskst{g}{e}$ determines a task to evaluate a goal $g$ in an environment $e$. For a conjunction node, its right child
is always a goal since it cannot be evaluated unless some result is provided by the left conjunct.
We also need a terminal state $\diamond$ to represent the end of the evaluation.
The label ``$\circ$'' is used to mark those steps which do not provide an answer; otherwise, a transition is labeled by an updated
environment.

\begin{figure*}[t]
  \renewcommand{\arraystretch}{1.6}
  \[
  \begin{array}{crcr}
    \multicolumn{3}{c}{\taskst{t_1 \equiv t_2}{(\sigma, n)} \xrightarrow{\circ} \Diamond , \, \, \nexists\; mgu\,(t_1 \sigma, t_2 \sigma)} &\ruleno{UnifyFail} \\
    \multicolumn{3}{c}{\taskst{t_1 \equiv t_2}{(\sigma, n)} \xrightarrow{(mgu\,(t_1 \sigma, t_2 \sigma) \circ \sigma),\, n)} \Diamond} & \ruleno{UnifySuccess} \\
    \multicolumn{3}{c}{\taskst{\mbox{\lstinline|fresh|}\, \ttt{x}\, .\, g}{(\sigma, n)} \xrightarrow{\circ} \taskst{g\,[\bigslant{\alpha_{n + 1}}{\ttt{x}}]}{( \sigma, n + 1)}} & \ruleno{Fresh} \\
    \multicolumn{3}{c}{\dfrac{R_i^{k_i}=\lambda\,\ttt{x}_1\dots \ttt{x}_{k_i}\,.\,g}{\taskst{R_i^{k_i}\,(t_1,\dots,t_{k_i})}{e} \xrightarrow{\circ} \taskst{g\,[\bigslant{t_1}{\ttt{x}_1}\dots\bigslant{t_{k_i}}{\ttt{x}_{k_i}}]}{e}}} & \ruleno{Invoke}\\[3mm]
    \taskst{g_1 \lor g_2}{e} \xrightarrow{\circ} \taskst{g_1}{e} \oplus \taskst{g_2}{e} & \ruleno{Disj} &
    \taskst{g_1 \land g_2}{e} \xrightarrow{\circ} \taskst{g_1}{e} \otimes g_2 & \ruleno{Conj} \\    
    \dfrac{s_1 \xrightarrow{l} \Diamond}{(s_1 \oplus s_2) \xrightarrow{l} s_2} & \ruleno{DisjStop} &
    \dfrac{s_1 \xrightarrow{l} s'_1}{(s_1 \oplus s_2) \xrightarrow{l} (s_2 \oplus s'_1)} &\ruleno{DisjStep} \\
    \dfrac{s \xrightarrow{\circ} \Diamond}{(s \otimes g) \xrightarrow{\circ} \Diamond} &\ruleno{ConjStop} &
    \dfrac{s \xrightarrow{e} \Diamond}{(s \otimes g) \xrightarrow{\circ} \taskst{g}{e}}  & \ruleno{ConjStopAns}\\
    \dfrac{s \xrightarrow{\circ} s'}{(s \otimes g) \xrightarrow{\circ} (s' \otimes g)} &\ruleno{ConjStep} &
    \dfrac{s \xrightarrow{e} s'}{(s \otimes g) \xrightarrow{\circ} (\taskst{g}{e} \oplus (s' \otimes g))} & \ruleno{ConjStepAns} 
  \end{array}
  \]
  \caption{Operational semantics of interleaving search}
  \label{fig:operanional_semantics_rules}
\end{figure*}

The transition rules are shown in Fig.~\ref{fig:operanional_semantics_rules}.
The first six rules define the evaluation of leaf states.
For the disjunction and conjunction, the corresponding node states are constructed.
For other types of goals the environment and the evaluated goal are updated in accordance with the task given by the goal: for an equality the most general unifier of the terms is incorporated into the substitution (or execution halts if the terms are non-unifiable); for a fresh construction a new variable is introduced and the counter of allocated variables is incremented; for a relational call the body of the relation is taken as the next goal.
The rest of the rules define composition of evaluation of substates for partial disjunctions and conjunctions.
For a partial disjunction, the first constituent is evaluated for one step, then the constituents are swapped (which constitutes the \emph{interleaving}), and the label is propagated.
When the evaluation of the first constituent of partial disjunction halts, the evaluation proceeds with the second constituent.
For a partial conjunction, the first constituent is evaluated until the answer is obtained, then the evaluation of the second constituent with this answer as the environment is scheduled for evaluation together with the remaining partial conjunction (via partial disjunction node).
When the evaluation of the first constituent of partial conjunction halts, the evaluation of the conjunction halts, too.

The introduced transition system is completely deterministic,
therefore a derivation sequence for a state $s$ determines a certain \emph{trace}~--- a sequence of states and labeled transitions between
them. It may be either finite (ending with the terminal state $\Diamond$) or infinite. We will denote by $\trs{s}$ the sequence of states in
the trace for initial state $s$ and by $\tra{s}$ the sequence of answers in the trace for initial state $s$. The sequence $\tra{s}$ corresponds
to the stream of answers in the reference \mK implementations.

In the following we rely on the following property of leaf states:

\begin{definition}
  A leaf state $\taskst{g}{\mkenv{\sigma}{n}}$ is well-formed iff $\fv{g}\cup\Dom\,(\sigma)\cup\VRan\,(\sigma)\subseteq\{\alpha_1,\dots,\alpha_n\}$, where
  $\fv{g}$ denotes the set of free variables in a goal $g$, $\Dom\,(\sigma)$ and $\VRan\,(\sigma)$~--- the domain of a substitution $\sigma$ and
  a set of all free variables in its image respectively.
\end{definition}

Informally, in a well-formed leaf state all free variables in goals and substitution respect the counter of free logical variables.
This definition is in fact an instance of a more general definition of well-formedness for all states, introduced in~\cite{CertifiedSemantics}, where it is
proven that the initial state is well-formed and any transition from a well-formed state results in a well-formed one.

Besides operational semantics, we will make use of a denotational one analogous to the least Herbrand model. For a relation $R^k$, its denotational semantics $\sembr{R^k}$ is
treated as a $k$-ary relation on the set of all ground terms, where each ``dimension'' corresponds to a certain argument of $R^k$. For example,
$\sembr{\lstinline|append$^o$|}$ is a set of all triplets of ground lists, in which the third component is a
concatenation of the first two. The concrete description of the denotational semantics is given in~\cite{CertifiedSemantics} as well as the proof of
the soundness and completeness of the operational semantics w.r.t. to the denotational one.

\begin{figure}[t]
\centering
\[
\begin{array}{ccll}
B_{nf} & = &  \unigoal{\mathcal{T}_\mathcal{X}}{\mathcal{T}_\mathcal{X}} \; \mid \;
                     \invokegoal{R^k}{\mathcal{T}_\mathcal{X}}{\mathcal{T}_\mathcal{X}} \\
C_{nf} & = & B_{nf} \; \mid \; \conjgoal{C_{nf}}{B_{nf}} \\
F_{nf} & = & C_{nf} \; \mid \; \freshgoal{X}{F_{nf} } \\
D_{nf} & = & F_{nf} \; \mid \; \disjgoal{D_{nf}}{F_{nf}}
\end{array}
\]
\caption{Disjunctive Normal Form for goals}
\label{fig:dnf}
\end{figure}

Finally, we explicitly enumerate all the restrictions required by our method to work:

\begin{itemize}
\item All relations have to be in DNF (set $D_{nf}$ in \figureword\ref{fig:dnf}). 
\item We only consider goals which converge with a finite number of answers.
\item All answers have to be ground (\emph{groundness} condition) for all relation invocations encountered
  during the evaluation.
\item All answers have to be unique (\emph{answer uniqueness} condition) for all relation invocations encountered
  during the evaluation.
\end{itemize}

\section{Scheduling Complexity}
\label{sec:scheduling}


We may notice that the operational semantics described in the previous section can be used to calculate the exact number of elementary scheduling steps.
Our first idea is to take the number of states $d\,(s)$ in the finite trace for a given state $s$:

\[ d\,(s) \; \eqdef \; | \trs{s} |  \]

However, it turns out that this value alone does not provide an accurate scheduling complexity estimation. The reason is that some
elementary steps in the semantics are not elementary in existing implementations. Namely, a careful analysis discovers that
each semantic step involves navigation to the leftmost leaf of the state which in implementations corresponds to multiple elementary actions,
whose number is proportional to the height of the leftmost branch of the state in question. Here we provide an \emph{ad-hoc} definition for this value, $t\,(s)$,
which we call the \emph{scheduling factor}:

\[
t\,(s) \eqdef \sum\limits_{s_i \in \trs{s}} lh\,(s_i) 
\]

where $lh\,(s_i)$ is the height of the leftmost branch of the state. 

In the rest of the section, we derive recurrent equations which would relate the scheduling complexity for states to the scheduling complexity for their
(immediate) substates. It turns out that to come up with such equations both $t$ and $d$ values have to be estimated simultaneously.



The next lemma provides the equations for $\oplus$-states:

\begin{replemma}{lem:sum_measure_equations}
For any two states $s_1$ and $s_2$

\[
\begin{array}{rcl}
  d\,(s_1 \oplus s_2) &=& d\,(s_1) + d\,(s_2) \\
    t\,(s_1 \oplus s_2) &=& t\,(s_1) + t\,(s_2) + \costdisj{s_1}{s_2}
\end{array}
\]

where $\costdisj{s_1}{s_2} = \min\,\{2\cdot d\,(s_1) - 1, 2\cdot d\,(s_2)\}$
\end{replemma}

Informally, for a state in the form $s_1 \oplus s_2$ the substates are evaluated separately, one step at a time for
each substate, so the total number of semantic steps is the sum of those for the substates. However, for the scheduling factor, 
there is an extra summand $\costdisj{s_1}{s_2}$ since the ``leftmost heights'' of the states in the trace are one node greater than those for the
original substates due to the introduction of one additional $\oplus$-node on the top. This additional node persists in the trace until the evaluation
of one of the substates comes to an end, so the scheduling factor is increased by the number of steps until that.

The next lemma provides the equations for $\otimes$-states:\footnote{We assume $\Diamond\otimes g = \Diamond$}

\begin{replemma}{lem:times_measure_equations}
  For any state $s$ and any goal $g$
  
\[
\begin{array}{rclr}
d\,(s \otimes g)  &=&  d\,(s) + \smashoperator[lr]{\sum\limits_{a_i \in \tra{s}}} d\,(\taskst{g}{a_i})& \qquad(\star) \\

 t\,(s \otimes g)  &=&  t\,(s) + \costconj{s}{g} + \smashoperator[lr]{\sum\limits_{a_i \in \tra{s}}} (t\,(\taskst{g}{a_i}) + \costdisj{\taskst{g}{a_i}}{(s'_i \otimes g)})&\qquad(\dagger)
\end{array}
\]

where 
\[
\begin{array}{rcl}
\costconj{s}{g} & = & d\,(s) \\
s'_i & = & \mbox{the first state in the trace for $s$ after} \\
 & & \mbox{a transition delivering the answer $a_i$} \\
\end{array}
\]
\end{replemma}

For the states of the form $s \otimes g$ the reasoning is the same, but the resulting equations are more complicated.
In an $\otimes$-state the left substate is evaluated until an answer is found, which is then taken as
\emph{an environment} for the evaluation of the right subgoal.
Thus, in the equations for $\otimes$-states the evaluation times of the second goal \emph{for all
the answers} generated for the first substate are summed up. The evaluation of the right subgoal
in different environments is added to the evaluation of the left substate via creation of
an $\oplus$-state, so for the scheduling factor there is
an additional summand $\costdisj{\taskst{g}{a_i}}{s'_i}$ for each answer with $s'_i$ being the state
after discovering the answer.
There is also an extra summand $\costconj{s}{g}$ for the scheduling factor because of the
$\otimes$-node that increases the height in the trace, analogous to the one caused by
$\oplus$-nodes.
Note, a $\otimes$-node is always placed immediately over the left substate so this
addition is exactly the number of steps for the left substate.

Unfolding costs definitions in $(\dagger)$ gives us a cumbersome formula that 
includes some intermediate states $s'_i$ encountered during the evaluation. However, as ultimately
we are interested in asymptotic estimations, we can approximate these costs up to a multiplicative constant. We can notice that the value $d\,(s'_i \otimes g)$ occurring in the second argument of $cost_{\oplus}$ includes values $d\,(\taskst{g}{a_j})$ (like in the first argument) for all answers $a_j$ after this intermediate state. So in the sum of all $cost_{\oplus}$ values $d\,(\taskst{g}{a_i})$ may be excluded for at most one answer, and in fact, if we take the maximal one of these values we will get a rather precise approximation. Specifically, we can state the following approximation\footnote{We assume the following definition for
$f\,(x) = g\,(x) + \Theta\,(h\,(x))$: \[\exists C_1, C_2 \in \mathcal{R^{+}}, \, \forall x : g\,(x) + C_1 \cdot h\,(x) \le f\,(x) \le g\,(x) + C_2 \cdot h\,(x) \]}
for $t\,(s \otimes g)$.

\begin{replemma}{lem:otimes_t_approximation}
\[
 t\,(s \otimes g)  =  t\,(s) + \left({\sum\limits_{a_i \in \tra{s}}} t\,(\taskst{g}{a_i})\right) +
 \Theta\,(d\,(s) + \smashoperator[lr]{\sum\limits_{a_i \in \tra{s}}} d\,(\taskst{g}{a_i}) - \smashoperator{\maxd\limits_{a_i \in \tra{s}}} d\,(\taskst{g}{a_i}))	
\]
\end{replemma}

Hereafter we use the following notation: $\displaystyle{\maxd}\, S = \max\,(S\cup\{0\})$.
We can see that the part under $\Theta$ is very similar to the $(\star)$ except that here we exclude $d$ value for one of the answers from the sum.
This difference is essential and, as we will see later, it is in fact responsible for the difference in complexities for our motivating example.

\section{Complexity Analysis via Symbolic Execution}
\label{sec:symbolic}

Our approach to complexity analysis is based on a semi-automatic procedure involving symbolic execution.
In the previous section, we presented formulae to compositionally estimate the complexity factors for
\emph{non-leaf states} of operational semantics under the assumption that corresponding estimations for
\emph{leaf states} are given. In order to obtain corresponding estimations for relations as a whole, we
would need to take into account the effects of relational invocations, including the recursive ones.

Another observation is that as a rule we are interested in complexity estimations in terms of some \emph{metatheory}. For
example, dealing with relations on lists we would be interested in estimations in terms of list lengths,
with trees~--- in terms of depth or number of nodes, with numbers~--- in terms of their values, etc. It is
unlikely that a generic term-based framework would provide such specific information automatically. Thus,
a viable approach would be to extract some inequalities involving the complexity factors of certain relational
calls automatically and then let a human being solve these inequalities in terms of a relevant metatheory.

For the sake of clarity we will provide a demonstration of complexity analysis for a specific example~---
\lstinline|append$^o$| relation from the introduction~--- throughout the section.

The extraction procedure utilizes a symbolic execution technique and is completely automatic.
It turns out that the semantics we have is already abstract enough to be used for symbolic
execution with minor adjustments.
In this symbolic procedure, we mark some of the logic variables as ``grounded'' and at certain
moments substitute them with ground terms.
Informally, for some goal with some free logic variables we consider the complexity of a search
procedure which finds the bindings for all non-grounded variables based on the ground values
substituted for the grounded ones.
This search procedure is defined precisely by the operational semantics; however, as the concrete
values of grounded variables are unknown (only the fact of their \emph{groundness}), the whole
procedure becomes symbolic.
In particular, in unification the groundness can be propagated to some non-grounded free variables.
Thus, the symbolic execution is determined by a set of grounded variables (hereafter denoted
as $V\subset\mathcal A$). The initial choice of $V$ determines the problem we analyze.

In our example the objective is to study the execution when we specialize the first two arguments with
ground values and leave the last argument free. Thus, we start with the goal \lstinline|append$^o$ $a$ $b$ $ab$|
(where $a$, $b$ and $ab$ are distinct free logic variables) and set the initial $V = \{ a, b \}$.

We can make an important observation that both complexity factors ($d$ and $t$) are stable w.r.t. the renaming of
free variables; moreover, they are also stable w.r.t. the change of the fresh variables counter as long as it stays
adequate, and change of current substitution, as long as it gives the same terms after application.
Formally, the following lemma holds.

\begin{replemma}{lem:measures_changing_env}
Let $s = \taskst{g}{\mkenv{\sigma}{n}}$ and $s' = \taskst{g^\prime}{\mkenv{\sigma^\prime}{n^\prime}}$ be two well-formed states.
If  there exists a bijective substitution $\pi \colon FV\,(g \sigma) \to FV\,(g^\prime \sigma^\prime)$ such that
$g \sigma \pi = g^\prime \sigma^\prime $, then $d\,(s) = d\,(s^\prime)$ and $t\,(s) = t\,(s^\prime)$.
\end{replemma}

The lemma shows that the set of states for which a call to relation has to be analyzed can
be narrowed down to a certain family of states. 

\begin{definition} Let $g$ be a goal. An initial state for $g$ is $init\,(g)=\taskst{g}{(\varepsilon, \ninit\,(g))} $
with $ n_{init}\,(g) = \min\, \{ n \mid FV\,(g) \subseteq \{ \alpha_1\dots\alpha_n \} \} $
\end{definition}

Due to the Lemma~\ref{lem:measures_changing_env} it is sufficient for analysis of a relational call to consider only the family of initial states since an arbitrary call state
encountered throughout the execution can be transformed into an initial one while preserving both
complexity factors. Thus, the analysis can be performed in a compositional manner where each call can be analyzed separately.
For our example the family of initial states is $q^{app}(\mathbf{a}, \mathbf{b}) = init\,(\lstinline|append$^o$ $\mathbf{a}$ $\mathbf{b}$ $ab$|)$ for arbitrary ground
terms $\mathbf{a}$ and $\mathbf{b}$.

As we are aiming at the complexity estimation depending on specific ground values substituted for grounded variables, in general case extracted
inequalities have to be parameterized by \emph{valuations}~--- mappings from the set of grounded variables to ground terms. As the new variables
are added to this set during the execution, the valuations need to be extended for these new variables. The following definition introduces this notion.

\begin{definition}
  Let $ V \subset U \subset \mathcal{A} $ and $ \rho \colon V \to \mathcal{T}_{\emptyset} $ and $ \rho^\prime \colon U \to \mathcal{T}_{\emptyset} $ be two valuations. We say that $\rho^\prime$ extends $\rho$ (denotation: $ \rho^\prime \succ \rho$) if $\rho^\prime\,(x) = \rho\,(x)$ for all $x \in V$.
\end{definition}

The main objective of the symbolic execution in our case is to find constraints on valuations for every leaf goal in the body of a relation that determine
whether the execution will continue and how a valuation changes after this goal. For internal relational calls, we describe constraints in terms of denotational
semantics (to be given some meaning in terms of metatheory later). We can do it because of a precise equivalence between the answers found by operational semantics
and values described by denotational semantics thanks to soundness and completeness as well as our requirements of grounding and uniqueness of answers.
Our symbolic treatment of equalities relies on the fact that substitutions of ground terms commute, in a certain sense, with unifications. More specifically,
we can use the most general unifier for two terms to see how unification goes for two terms with some free variables substituted with ground terms.
The most general unifier may contain bindings for both grounded and non-grounded variables. A potential most general unifier for terms after substitution
contains the same bindings for non-grounding terms (with valuation applied to their rhs), while bindings for grounding variables turn into equations that
should be satisfied by the unifier with ground value on the left and bound term on the right. In particular, this means that all variables in bindings for
grounded variables become grounded, too. We can use this observation to define an iterative process that determines the updated set of grounded
variables $\upd{U}{\delta}$ for a current set $U$ and a most general unifier $\delta$ and a set of equations $\constr{\delta}{U}$ that should be
respected by the valuation.

\[
\begin{array}{rcl}
\upd{U}{\delta} &=& \begin{cases}
                           U & \quad\forall x \in U : FV\,(\delta\,(x)) \subset U \\
                           \upd{U \cup \displaystyle\bigcup\limits_{x \in U} FV\,(\delta\,(x))}{\delta} & \quad\mbox{otherwise}
                          \end{cases}\\
\constr{\delta}{U} &=& \{ x = \delta\,(x) \mid x \in U \cap \mathcal{D}om\,(\delta) \}
\end{array}
\]

Using these definitions we can describe symbolic unification by the following lemma.

\begin{replemma}{lem:symbolic_unification_soundness}
Let $t_1$, $t_2$ be terms,  $V \subset \mathcal{A}$ and $\rho \colon V \to \grterms$ be a valuation. If $mgu\,(t_1, t_2) = \delta$ and $U = \upd{V}{\delta} $  then $t_1 \rho$ and $t_2 \rho$ are unifiable iff there is some $\rho' \colon U \to \grterms$ such that $\rho' \succ \rho$ and $\forall (y = t) \in \constr{\delta}{U}\,:\, \rho'(y) = t \rho'$.
In such case $\rho'$ is unique and $ \rho \circ mgu\,(t_1 \rho, t_2 \rho) = \delta\circ\rho' $ up to alpha-equivalence (e.g. there exists a bijective substitution $\pi : FV(t_1) \to FV(t_2)$, s.t. $ \rho \circ mgu\,(t_1 \rho, t_2 \rho) = \delta \circ\rho'\circ \pi$).
\end{replemma}

\begin{figure}[t]
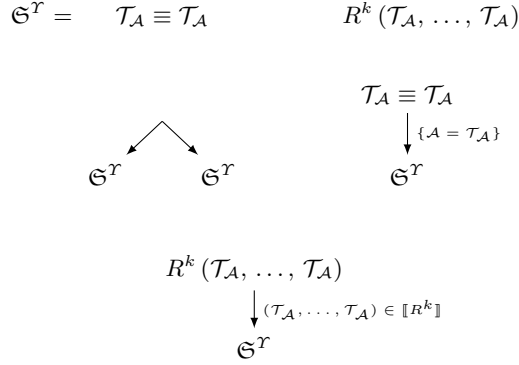

\centering
\[\renewcommand{\arraystretch}{1}
\begin{array}{cccp{1cm}c}
  \schemewithvset{\mathfrak{S}}{\Upsilon} & = &\schemenode{$\unigoal{\mathcal{T}_\mathcal{A}}{\mathcal{T}_\mathcal{A}}$} & &\schemenode{$\invokegoal{R^k}{\mathcal{T}_\mathcal{A}}{\mathcal{T}_\mathcal{A}}$} \\[6mm] 
                                          &   &\schemefork{$\schemewithvset{\mathfrak{S}}{\Upsilon}$}{$\schemewithvset{\mathfrak{S}}{\Upsilon}$} &&
                                               \schemesarrow{$\unigoal{\mathcal{T}_\mathcal{A}}{\mathcal{T}_\mathcal{A}}$}{$\{\mathcal{A}=\mathcal{T}_\mathcal{A}\}$}{$\schemewithvset{\mathfrak{S}}{\Upsilon}$}\\[6mm]  
                                          &   &\multicolumn{3}{c}{\schemedarrow{$\invokegoal{R^k}{\mathcal{T}_\mathcal{A}}{\mathcal{T}_\mathcal{A}}$}{$(\mathcal{T}_\mathcal{A}, \dots, \mathcal{T}_\mathcal{A}) \in \sembr{R^k}$}{$\schemewithvset{\mathfrak{S}}{\Upsilon}$}}
\end{array}\]
\caption{Symbolic Scheme Forms}
\label{fig:scheme_fragments}
\end{figure}

In our description of the extraction process, we use a visual representation of symbolic execution of a relation body for a given set of grounded variables in a form of a \emph{symbolic scheme}.
A symbolic scheme is a tree-like structure with different branches corresponding to execution of different disjuncts and nodes corresponding to equalities and relational calls in the body
augmented with subsets of grounded variables at the point of execution.\footnote{Note the difference with conventional symbolic execution graphs
with different branches representing mutually exclusive paths of evaluation, not the different parts within one evaluation.} Constraints for substituted
grounded variables that determine whether the execution continues are presented as labels on the edges of a scheme.

Each scheme is built as a composition of the five patterns, shown in \figureword~\ref{fig:scheme_fragments} (all schemes are indexed by subsets of grounded variables with $\Upsilon = 2^{\mathcal{A}}$ denoting such subsets).

Note, the constraints after nodes of different types differ: unification puts a constraint in a form of a set of equations on substituted ground values that should be respected while
relational call puts a constraint in a form of a tuple of ground terms that should belong to the denotational semantics of a relation.

The construction of a scheme for a given goal (initially, the body of a relation) mimics a regular execution of a relational program. The derivation rules for scheme
formation have the following form $\schemetrans{g}{\Gamma}{\sigma}{n}{V}{\schemewithvset{\mathfrak{S}}{V}}$. Here $g$ is a goal, $\Gamma$ is a list of \emph{deferred}
goals (these goals have to be executed after the execution of $g$ in every branch in the same order,
initially this list is empty; this resembles continuations, but the analogy is not complete), $\sigma$ and $n$ are substitution and counter from the current
environment respectively, $V$ is a set of grounded variables at the moment.

\begin{figure}[t]
\renewcommand{\arraystretch}{2}
  \[
\begin{array}{cr}
  \onepremrule
		{  \schemetrans{g_1}{g_2 : \Gamma}{\sigma}{n}{V}{\schemewithvset{\mathfrak{S}}{V}}  } 
		{  \schemetrans{\conjgoal{g_1}{g_2}}{\Gamma}{\sigma}{n}{V}{ \schemewithvset{\mathfrak{S}}{V} }  } & \ruleno{Conj$_\mathfrak S$}
		\\[5mm]
                
  \twopremrule
		{  \schemetrans{g_1}{\Gamma}{\sigma}{n}{V}{\schemewithvset{\mathfrak{S_1}}{V}}  }
		{  \schemetrans{g_2}{\Gamma}{\sigma}{n}{V}{\schemewithvset{\mathfrak{S_2}}{V}}  }
		{  \schemetrans{\disjgoal{g_1}{g_2}}{\Gamma}{\sigma}{n}{V}{\parbox[m]{2cm}{ \schemefork{$\schemewithvset{\mathfrak{S_1}}{V}$}{$\schemewithvset{\mathfrak{S_2}}{V}$}} }  } & \ruleno{Disj$_\mathfrak S$}\\[10mm]

 \onepremrule
		{  \schemetrans{\substitute{g}{\alpha_n}{\ttt{x}}}{\Gamma}{\sigma}{n + 1}{V}{\schemewithvset{\mathfrak{S}}{V}}  }
		{  \schemetrans{\freshgoal{\ttt{x}}{g}}{\Gamma}{\sigma}{n}{V}{ \schemewithvset{\mathfrak{S}}{V} }  } & \ruleno{Fresh$_\mathfrak S$}\\[5mm]

  \schemetrans{\unigoal{t_1}{t_2}}{\epsilon}{\sigma}{n}{V}{\parbox[m]{2cm}{\schemenode{$\unigoal{t_1 \sigma}{t_2 \sigma}$}}}& \ruleno{UnifyLeaf$_\mathfrak S$}\\[5mm]

 \schemetrans{\invokegoal{R^k}{t_1}{t_k}}{\epsilon}{\sigma}{n}{V}{\parbox[m]{2cm}{\schemenode{$\invokegoal{R^k}{t_1 \sigma}{t_k \sigma}$}}}& \ruleno{InvokeLeaf$_\mathfrak S$}\\[5mm] 
		
  \onepremrule
		{  \nexists mgu\,(t_1 \sigma, t_2 \sigma)  }
		{  \schemetrans{\unigoal{t_1}{t_2}}{g : \Gamma}{\sigma}{n}{V}{\parbox[m]{2cm}{\schemenode{$\unigoal{t_1 \sigma}{t_2 \sigma}$}}} }& \ruleno{UnifyFail$_\mathfrak S$}\\[5mm]

  \threepremrule
		{  mgu\,(t_1 \sigma, t_2 \sigma) = \delta  }
		{  U = \upd{V}{\delta}  }
		{  \schemetrans{g}{\Gamma}{\sigma \delta}{n}{U}{\schemewithvset{\mathfrak{S}}{U}}  }
		{  \schemetrans{\unigoal{t_1}{t_2}}{g : \Gamma}{\sigma}{n}{V}{\parbox[m]{2cm}{\schemesarrow{$\unigoal{t_1 \sigma}{t_2 \sigma}$}{$\constr{\delta}{U}$}{$\schemewithvset{\mathfrak{S}}{U}$}} }   } & \ruleno{UnifySuccess$_\mathfrak S$}\\[15mm]
		
  \twopremrule
		{  \mbox{\phantom{XXXXXX}} U =  V \cup \displaystyle\bigcup\limits_{i} FV\,(t_i \sigma) }
		{  \schemetrans{g}{\Gamma}{\sigma}{n}{U}{\schemewithvset{\mathfrak{S}}{U}} \mbox{\phantom{XXXXXX}} }
		{  \schemetrans{\invokegoal{R^k}{t_1}{t_k}}{g : \Gamma}{\sigma}{n}{V}{ \parbox[m]{2cm}{\schemedarrow{$\invokegoal{R^k}{t_1 \sigma}{t_k \sigma}$}{$ (t_1 \sigma, \dots, t_k \sigma) \in \sembr{R^k} $}{$\schemewithvset{\mathfrak{S}}{U}$}} }   } & \ruleno{Invoke$_\mathfrak S$}
 \end{array}
\]
\caption{Scheme Formation Rules}
\label{fig:scheme_formation}
\end{figure}

The rules are shown in \figureword~\ref{fig:scheme_formation}. \ruleno{Conj$_\mathfrak S$} and \ruleno{Disj$_\mathfrak S$} are structural rules: when investigating conjunctions we defer
the second conjunct by adding it to $\Gamma$ and continue with the first conjunct; disjunctions simply result in forks. \ruleno{Fresh$_\mathfrak S$} introduces a fresh logic
variable (not grounded) and updates the counter of occupied variables. When the investigated goal is equality or relational call it is added as a node to the scheme. If there are
no deferred goals, then this node is a leaf (rules \ruleno{UnifyLeaf$_\mathfrak S$} and \ruleno{InvokeLeaf$_\mathfrak S$}). Equality is also added as a leaf if there are some deferred goals,
but the terms are non-unifiable and so the execution stops (rule \ruleno{UnifyFail$_\mathfrak S$}). If the terms in the equality are unifiable and there are deferred goals
(rule \ruleno{UnifySuccess$_\mathfrak S$}), the equality is added as a node and the execution continues for the deferred goals, starting from the leftmost one; also the set of grounded variables
is updated and constraint labels are added for the edge in accordance with \lemmaword~\ref{lem:symbolic_unification_soundness}. For relational calls the proccess is similar: if there are some deferred goals
(rule \ruleno{Invoke$_\mathfrak S$}), all variables occurring in a call become grounded (due to the grounding condition we imposed) and should satisfy the denotational semantics
of the invoked relation.

The scheme constructed by these rules for our \lstinline|append$^o$| example is shown in \figureword~\ref{fig:example_scheme}. For simplicity, we do not show the set of grounded
variables for each node, but instead overline grounded variables in-place. Note, all variables that occur in constraints on the edges are grounded after parent node execution.

Now, we can use schemes to see how the information for leaf goals in relation body is combined with conjunctions and disjunctions. Then we can apply formulae
from \sectionword~\ref{sec:scheduling} to get recursive inequalities (providing lower and upper bounds simultaneously) for both complexity factors.

In these inequalities we need to sum up the values of $d$ and $t$-factor for all leaf goals of a body and for all environments these goals are evaluated for.
The leaf goals are
the nodes of the scheme and evaluated environments can be derived from the constraints attached to the edges. So, for this summation we introduce the following notions: $\mathcal{D}$
is the sum of $d$-factor values and $\mathcal{T}$ is the sum of $t$-factor values for the execution of the body with specific valuation $\rho$.

\begin{wrapfigure}{r}{0.5\textwidth}
\begin{center}
\begin{tikzpicture}[level distance=30pt, sibling distance=10em, edge from parent/.style={draw,-latex}]
   \coordinate   
      child { node {$\unigoal{\overline{a}}{\texttt{Nil}}$}
        child { node {$\unigoal{ab}{\overline{b}}$}
                  edge from parent node[right]{\tiny{${a} = \texttt{Nil}$}} } }
      child { node {$\unigoal{\overline{a}}{\texttt{Cons($h$, $t$)}}$} 
      	child { node {\texttt{append$^o$ $\overline{t}$ $\overline{b}$ $tb$}}
      	   child { node {$\unigoal{ab}{\texttt{Cons($\overline{h}$, $\overline{tb}$)}}$}
      	             edge from parent node[right]{\tiny{$({t}, {b}, {tb}) \in \llbracket \texttt{append$^o$} \rrbracket$}}  }
      	   edge from parent node[right]{\tiny{${a} = \texttt{Cons(${h}$, ${t}$)}$}}  } } ;
\end{tikzpicture}
\end{center}
\caption{Symbolic execution scheme for the goal  \lstinline|append$^o$ $\,a\;$ $b\;$ $ab$|  with initial set of grounded variables $V = \{ a, b \}$. For each node, variables that
  are grounded at the point of execution of this node are overlined. }
\label{fig:example_scheme}
\end{wrapfigure}
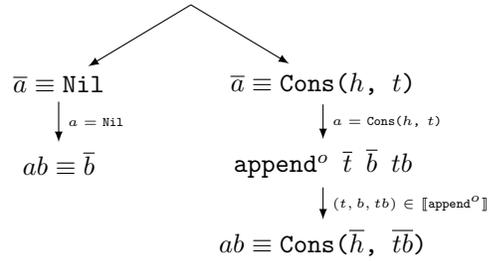

 Their definitions are shown in
\figureword~\ref{fig:scheduling_extraction_d_t} (both formulas are given in the same figure as the definitions coincide modulo factor denotations). For nodes, we take
the corresponding value (for equality it always equals $1$). When going through an equality we sum up the rest with an updated valuation (by \lemmaword~\ref{lem:symbolic_unification_soundness}
this sum always has one or zero summands depending on whether the unification succeeds or not). When going through a relational call we take a sum of all valuations that satisfy the
denotational semantics (these valuations will correspond exactly to the set of all answers produced by the call since operational semantics is sound and complete w.r.t. the denotational
one and because we require all answers to be unique). For disjunctions, we take the sum of both branches.

\begin{figure}[t]
\[
\begin{array}{rclcl}
 \mathcal{\nicefrac{D}{T}}\,(&\parbox[m]{1.3cm}{\schemenode{$\unigoal{t_2}{t_2}$}}&)(\rho) &=& 1  \\

 \mathcal{\nicefrac{D}{T}}\,(&\parbox[m]{2.5cm}{\schemenode{$\invokegoal{R^k}{t_1}{t_k}$}}&)(\rho) &=& \nicefrac{d}{t}\,(init\,(\invokegoal{R^k}{t_1 \rho}{t_k \rho})) \\

 \mathcal{\nicefrac{D}{T}}\,(&\parbox[m]{2cm}{\schemesarrow{$\unigoal{t_1}{t_2}$}{$Cs$}{$\schemewithvset{\mathfrak{S}}{U}$}} &)(\rho) &=& 1 +
      \smashoperator{\sum\limits_{\substack{ \rho' \colon V \to \grterms \\
                                      \rho' \succ \rho \\
                                      \forall (y, t) \in Cs\,:\, \rho'\,(y) = t\, \rho'  }}}
           \mathcal{\nicefrac{D}{T}}\,(\schemewithvset{\mathfrak{S}}{U})(\rho')  \\

 \mathcal{\nicefrac{D}{T}}\,(& \parbox[m]{4cm}{\schemedarrow{$\invokegoal{R^k}{t_1}{t_k}$}{ $(t_1, \dots, t_k) \in \sembr{R^k}  $}{$\schemewithvset{\mathfrak{S}}{U}$}} &)(\rho) &=&
      \nicefrac{d}{t}\,(init\,(\invokegoal{R^k}{t_1 \rho}{t_k \rho})) +
      \smashoperator{\sum\limits_{\substack{ \rho' \colon V \to \grterms \\
                                      \rho' \succ \rho \\
                                      (t_1 \rho', \dots, t_k \rho') \in \sembr{R^k}  }}}
           \mathcal{\nicefrac{D}{T}}\,(\schemewithvset{\mathfrak{S}}{U})(\rho')  \\

 \mathcal{\nicefrac{D}{T}}\,(&\parbox[m]{2.5cm}{\schemefork{$\schemewithvset{\mathfrak{S}_1}{V}$}{$\schemewithvset{\mathfrak{S}_2}{V}$}}&)(\rho) &=&
 \mathcal{\nicefrac{D}{T}}\,(\schemewithvset{\mathfrak{S}_1}{V})(\rho) + \mathcal{\nicefrac{D}{T}}\,(\schemewithvset{\mathfrak{S}_2}{V})(\rho)
\end{array}
\]
\caption{Complexity Factors Extraction: $\mathcal D$ and $\mathcal T$}
\vspace{5mm}
\label{fig:scheduling_extraction_d_t}
\end{figure}

\begin{figure}[t]
\[
\begin{array}{rclcl}
 \mathcal{L}\,(&\parbox[m]{1.3cm}{\schemenode{$\unigoal{t_2}{t_2}$}}&)(\rho) &=& \{init\,(\unigoal{t_2}{t_2})\} \\
 \mathcal{L}\,(&\parbox[m]{2.5cm}{\schemenode{$\invokegoal{R^k}{t_1}{t_k}$}}&)(\rho) &=& \{init\,(\invokegoal{R^k}{t_1 \rho}{t_k \rho})\} \\
 \mathcal{L}\,(&\parbox[m]{2cm}{\schemesarrow{$\unigoal{t_1}{t_2}$}{$Cs$}{$\schemewithvset{\mathfrak{S}}{U}$}} &)(\rho) &=&  \smashoperator{\bigcup\limits_{\substack{ \rho' \colon V \to \grterms \\
                                      \rho' \succ \rho \\
                                      \forall (y, t) \in Cs\,:\, \rho'\,(y) = t\, \rho'  }}}
           \mathcal{L}\,(\schemewithvset{\mathfrak{S}}{U})(\rho')  \\
 \mathcal{L}\,(& \parbox[m]{4cm}{\schemedarrow{$\invokegoal{R^k}{t_1}{t_k}$}{ $(t_1, \dots, t_k) \in \sembr{R^k}  $}{$\schemewithvset{\mathfrak{S}}{U}$}} &)(\rho) &=&
      \smashoperator{\bigcup\limits_{\substack{ \rho' \colon V \to \grterms \\
                                      \rho' \succ \rho \\
                                      (t_1 \rho', \dots, t_k \rho') \in \sembr{R^k}  }}}
           \mathcal{L}\,(\schemewithvset{\mathfrak{S}}{U})(\rho')  \\
 \mathcal{L}\,(&\parbox[m]{2.5cm}{\schemefork{$\schemewithvset{\mathfrak{S}_1}{V}$}{$\schemewithvset{\mathfrak{S}_2}{V}$}}&)(\rho) &=&
 \mathcal{L}\,(\schemewithvset{\mathfrak{S}_1}{V})(\rho) \cup \mathcal{L}\,(\schemewithvset{\mathfrak{S}_2}{V})(\rho)
\end{array}
\]
\caption{Complexity Factors Extraction: $\mathcal L$}
\label{fig:scheduling_extraction_l}
\end{figure}

As we saw in \sectionword~\ref{sec:scheduling} when computing the scheduling factors we need to exclude from the additional cost the value of $d$-factor
for one of the environments (the largest one). This is true for the generalized formula for a whole scheme, too. This time we need to take all executed
environments for all the leaves of a scheme and exclude the $d$-factor value for a maximal one (the formula for conjunction ensures that we make the exclusion
for the leaf, and the formula for disjunction ensures that we make it for only one of the leaves). So, we will need additional notion $\mathcal{L}$,
similar to $\mathcal{D}$ and $\mathcal{T}$ that will collect all the goals of the form $init\,(g_i \rho)$, where $g_i$ is a leaf goal and $\rho$ is a
valuation corresponding to one of the environments this leaf is evaluated for. The definition of  $\mathcal{L}$ is shown in \figureword~\ref{fig:scheduling_extraction_l}.

Now we can formulate the following main theorem that provides the principal recursive inequalities, extracted from the scheme for a given goal.

\begin{reptheorem}{extracted_approximations}
Let $g$ be a goal, and let $\schemetrans{g}{\epsilon}{\varepsilon}{n_{init}(g)}{V}{\schemewithvset{\mathfrak{S}}{V}}$. Then

\[
\begin{array}{rcl}
    d\,(init\,(g\,\rho)) &=& \mathcal{D}\,(\schemewithvset{\mathfrak{S}}{V})(\rho) + \Theta\,(1) \\
   t\,(init\,(g\,\rho)) &=& \mathcal{T}\,(\schemewithvset{\mathfrak{S}}{V})(\rho) + \Theta\,(\mathcal{D}\,(\schemewithvset{\mathfrak{S}}{V})(\rho)
   - \smashoperator{\maxd\limits_{\taskst{g_i}{e_i} \in \mathcal{L}(\schemewithvset{\mathfrak{S}}{V})(\rho)}} d\,(\taskst{g_i}{e_i}) + 1)
\end{array}
\]

\noindent being considered as functions on $\rho \colon V \to T_{\emptyset}$
\end{reptheorem}

The theorem allows us to extract two inequalities (upper and lower bounds) for both factors with a multiplicative constant that is the same for all valuations.

For our example, we can extract the following recursive inequalities from the scheme in \figureword~\ref{fig:example_scheme}. For presentation purposes, we will
not show valuation in inequalities explicitly, but instead show the ground values of grounded variables (using variables in bold font) that determine each valuation.
We can do such a simplification for any concrete relation.

\[
d\,(q^{app}\,(\mathbf{a}, \mathbf{b}))  = (1 + \sum\limits_{\mathbf{a} = \texttt{Nil}} 1) + (1 + \smashoperator{\sum\limits_{\mathbf{h}, \mathbf{t}: \mathbf{a} = \texttt{Cons($\mathbf{h}$, $\mathbf{t}$)}}} (d\,(q^{app}\,(\mathbf{t}, \mathbf{b})) + \smashoperator{\sum\limits_{\mathbf{tb} : (\mathbf{t}, \mathbf{b}, \mathbf{tb}) \in \llbracket \texttt{append$^o$} \rrbracket}} 1)) + \Theta\,(1)
\]\\[0.8mm]
\[
\begin{array}{lclc}
t\,(q^{app}\,(\mathbf{a}, \mathbf{b})) & = & (1 + \sum\limits_{\mathbf{a} = \texttt{Nil}} 1) + (1 + \smashoperator{\sum\limits_{\mathbf{h}, \mathbf{t}: \mathbf{a} = \texttt{Cons($\mathbf{h}$, $\mathbf{t}$)}}} (t\,(q^{app}\,(\mathbf{t}, \mathbf{b})) + \smashoperator{\sum\limits_{\mathbf{tb} : (\mathbf{t}, \mathbf{b}, \mathbf{tb}) \in \llbracket \texttt{append$^o$} \rrbracket}} 1)) & + \\
 & &\Theta((1 + \sum\limits_{\mathbf{a} = \texttt{Nil}} 1) + (1 + \smashoperator{\sum\limits_{\mathbf{h}, \mathbf{t}: \mathbf{a} = \texttt{Cons($\mathbf{h}$, $\mathbf{t}$)}}} (d\,(q^{app}\,(\mathbf{t}, \mathbf{b})) + \smashoperator{\sum\limits_{\mathbf{tb} : (\mathbf{t}, \mathbf{b}, \mathbf{tb}) \in \llbracket \texttt{append$^o$} \rrbracket}} 1)) & - \\
& & \smashoperator{\maxd\limits_{\substack{
                                   \mathbf{h}, \mathbf{t}, \mathbf{tb}: \mathbf{a} = \texttt{Cons($\mathbf{h}$, $\mathbf{t}$)} \land \\
                                   (\mathbf{t}, \mathbf{b}, \mathbf{tb}) \in \llbracket \texttt{append$^o$} \rrbracket
                                 }
                                }} \{ d\,(init\,(\unigoal{ab}{\mathbf{b}})), d\,(init\,(\unigoal{ab}{\texttt{Cons($\mathbf{h}$, $\mathbf{tb}$)}})) \} & + 1) 
\end{array}
\]

Automatically extracted recursive inequalities, as a rule, are cumbersome, but they contain all the information on how scheduling affects the complexity.
Often they can be drastically simplified by using metatheory-level reasoning.

For our example, we are only interested in the case when substituted values represent some lists. We thus perform the usual for lists case analysis
considering the first list empty or non-empty. We can also notice that the excluded summand equals one. So we can rewrite the inequalities in the following way:

\[
\begin{array}{lcl}
d\,(q^{app}\,(\texttt{Nil}, \mathbf{b})) & = & \Theta\,(1) \\
d\,(q^{app}\,(\texttt{Cons($\mathbf{h}$, $\mathbf{t}$)}, \mathbf{b})) & = & d\,(q^{app}\,(\mathbf{t}, \mathbf{b})) + \Theta\,(1) \\
t\,(q^{app}\,(\texttt{Nil}, \mathbf{b})) & = & \Theta\,(1) \\
t\,(q^{app}\,(\texttt{Cons($\mathbf{h}$, $\mathbf{t}$)}, \mathbf{b})) & = & t\,(q^{app}\,(\mathbf{t}, \mathbf{b})) + \Theta\,(d\,(q^{app}\,(\mathbf{t}, \mathbf{b}))) \\
\end{array}
 \]
 
These trivial linear inequalities can be easily solved:

\[
\begin{array}{lcl}
d\,(q^{app}\,(\mathbf{a}, \mathbf{b})) & = & \Theta\,(len\,(\mathbf{a})) \\
t\,(q^{app}\,(\mathbf{a}, \mathbf{b})) & = & \Theta\,(len^2\,(\mathbf{a})) \\
\end{array}
 \]
 
In this case, scheduling makes a big difference and changes the asymptotics. Note, we expressed the result using notions from metatheory
($len$ for the length of the list represented by a term).

In contrast, if we consider the optimal definition \lstinline|append$_{opt}^o$| the analysis of the call
$q^{app\text{-}opt}\,(\mathbf{a}, \mathbf{b}) = init\,(\texttt{append$_{opt}^o$} \, \mathbf{a} \, \mathbf{b} \, ab)$ is analogous,
but among the candidates for exclusion there is the value $d\,(q^{app\text{-}opt}\,(\mathbf{t}, \mathbf{b}))$ since the recursive
call is placed in a leaf. So the last simplified recursive approximation is the following (the rest is
the same as in our main example):

\[t\,(q^{app\text{-}opt}\,(\texttt{Cons\,($\mathbf{h}$, $\mathbf{t}$)}, \mathbf{b})) = t\,(q^{app\text{-}opt}\,(\mathbf{t}, \mathbf{b})) + \Theta\,(1) \]

So in this case the complexity of both factors is linear on $len\,(\mathbf{a})$.

\section{Evaluation}
\label{sec:evaluation}

The theory we have built was so far applied to only one relation~--- \lstinline|append$^o$|~--- which we used as a motivating example. With our framework, it turned out
to be possible to explain the difference in performance between two nearly identical implementations, and the difference~--- linear vs. quadratic asymptotic
complexity~--- was just as expected from the experimental performance evaluation. In this section, we present some other results of complexity estimations and
discuss the adequacy of these estimations w.r.t. the real \mK implementations.

Derived complexity estimations for a few other relations are shown in \figureword~\ref{fig:examples_all_complexities}. Besides concatenation, we deal with
naive list reversing and Peano numbers addition and multiplication. We show both $d-$ and $t-$ factors since the difference between the two indicates the
cases when scheduling strikes in. We expect that for simple relations like those presented the procedure of deriving estimations should be easy; however,
for more complex ones the dealing with extracted inequalities may involve a non-trivial metatheory reasoning.

The justification of the adequacy of our complexity estimations w.r.t. the existing \mK implementations faces the following problem: it is
not an easy task to separate the contribution of scheduling from other components of the search procedure~--- unification and occurs check. However, it is common
knowledge among \textsc{Prolog} users that in practice unification takes a constant time almost always; some theoretical basis for this is given in~\cite{UnificationAverageCost}.
There are some specifics of unification implementation in \mK. First, for the simplicity of backtracking in a non-mutable fashion triangular
substitution~\cite{UnificationTheory} is used instead of the idempotent one. It brings in an additional overhead which is analyzed
in some detail in~\cite{WillsThesis}, but the experience shows that in the majority of practical cases this overhead is insignificant. Second, \mK by default
performs the ``occurs check'', which contributes a significant overhead and often subsumes the complexity of all other search components. Meanwhile, it is known, that occurs checks are rarely
violated~\cite{OccursCheckNotAProblem}. Having said this, we expect that in the majority of the cases the performance of \mK programs with the occurs check disabled are
described by scheduling complexity alone. In particular, this is true for all cases in \figureword~\ref{fig:examples_all_complexities}. To confirm the adequacy of
our model we evaluated the running time of these and some other goals (under the conditions we've mentioned) and found that it confirms the estimations derived
using our framework. The details of implementation, evaluation, and results can be found in an accompanying repository.~\footnote{https://www.dropbox.com/sh/ciceovnogkeeibz/AAAoclpTSDeY3OMagOBJHNiSa?dl=0}

\begin{figure}[t]
  \setlength{\belowcaptionskip}{-20pt plus 3pt minus 2pt}
  \[ \begin{tabular}{|l|c|c||l|c|c|}
    \hline
          & $d$ & $t$  & & $d$ & $t$ \\
    \hline
    \lstinline|append$^o$ $\mathbf{a}$ $\mathbf{b}$ ${ab}$|      & $len\,(\mathbf{a})$ & $len^2\,(\mathbf{a})$   & \lstinline|plus$^o$ $\mathbf{n}$ $\mathbf{m}$ ${r}$| & $|\mathbf{n}|$ & $|\mathbf{n}|$ \\
    \lstinline|append$_{opt}^o$ $\mathbf{a}$ $\mathbf{b}$ ${ab}$| & $len\,(\mathbf{a})$ & $len\,(\mathbf{a})$     & \lstinline|plus$^o$ $\mathbf{n}$ ${m}$ $\mathbf{r}$| & $\min\,\{|\mathbf{n}|, |\mathbf{r}|\}$ & $\min\,\{|\mathbf{n}|, |\mathbf{r}|\}$ \\
    \lstinline|append$_{opt}^o$ ${a}$ ${b}$ $\mathbf{ab}$|        & $len\,(\mathbf{ab})$ & $len\,(\mathbf{ab})$   & \lstinline|plus$^o$ ${n}$ ${m}$ $\mathbf{r}$| & $|\mathbf{r}|$ & $|\mathbf{r}|$ \\
    \lstinline|revers$^o$ $\mathbf{a}$ ${r}$|                    & $len^2\,(\mathbf{a})$ & $len^3\,(\mathbf{a})$ & \lstinline|mult$^o$ $\mathbf{n}$ $\mathbf{m}$ ${r}$| & $|\mathbf{n}|\cdot|\mathbf{m}|$ & $|\mathbf{n}|^2\cdot|\mathbf{m}|$ \\
    \lstinline|revers$^o$ ${a}$ $\mathbf{r}$|                    & $len^2\,(\mathbf{r})$ & $len^2\,(\mathbf{r})$ & \lstinline|mult$^o$ (S ${n}$) (S ${m}$) $\mathbf{r}$| & $|\mathbf{r}|^2$ & $|\mathbf{r}|^2$ \\
    \hline
  \end{tabular} \]
  \caption{Derived $d$- and $t$-factors for some goals; $len\,(\bullet)$ stands for the length of a ground list, $|\bullet|$~--- for the value of a Peano number, represented as ground term.}
  \label{fig:examples_all_complexities}
\end{figure}

\section{Related Work}
\label{sec:related}

To our knowledge, our work is the first attempt of comprehensive time complexity analysis for interleaving search in \mK.
There is a number of separate observations on how certain patterns in relational programming affect performance and a number
of ``rules of thumb'' based on these observations~\cite{WillsThesis}. Some papers~\cite{GuardedFresh, FloatArithmetics, UniversalQuantification} tackle
specific problems with relational programming using quantitative time measuring for evaluation. These approaches to performance analysis, being sufficient
for specific relational problems, do not provide the general understanding of interleaving search and its cost.

At the same time complexity analysis was studied extensively in the broader context of logic programming (primarily, for \textsc{Prolog}).
As one important motivation for complexity analysis is granularity control in parallel execution, the main focus was set on the automated approaches.

Probably the best known among them is the framework~\cite{CostAnalysisLP} for cost analysis of logic programs (demonstrated on plain \textsc{Prolog}),
implemented in the system \textsc{CASLOG}. It uses data dependency information to estimate the sizes of the arguments and the number of solutions for
executed atoms. These estimations are formulated as recursive inequalities (more precisely, as difference equations for upper bounds), which are then
automatically solved with known methods. The time and space complexity are expressed using these estimations, the variations of this approach can
provide both upper~\cite{CostAnalysisLP} and lower~\cite{CostAnalysisLPLower} bounds.

An alternative approach is suggested in~\cite{SymbolicAnalysisLP} for symbolic analysis of logic programs (demonstrated in \textsc{Prolog} with cuts).
It constructs symbolic evaluation graphs capturing grounding propagation and reduction of recursive calls to previous ones, in the process of construction
some heuristic approximations are used. These graphs may look similar to the symbolic schemes described in the \sectionword~\ref{sec:symbolic} at first glance,
but there is a principal difference: symbolic graphs capture the whole execution with all invoked calls (using inverse edges to represent cycles with recursive calls),
while our schemes capture only the execution inside the body of a specific relation (representing the information about internal calls in terms of denotational semantics).
The graphs are then transformed into term rewriting systems, for which the problem of the complexity analysis is well-studied (specifically, \textsc{AProVE} tool is used).

While these two approaches can be seen as partial bases for our technique, they are mainly focused on how the information about the arguments and results of
the evaluated clauses can be derived automatically, since the calculation of time complexity of \textsc{SLD}-resolution is trivial when this information is available.
In contrast, we are interested in the penalty of non-trivial scheduling of relational calls under interleaving search, so we delegate handling the information
about calls to the reasoning in terms of a specific metatheory.

\section{Discussion and Future Work}
\label{sec:discussion}

The formal framework presented in this paper analyzes the basic aspects of scheduling cost for interleaving search strategy from the theoretical viewpoint.
As we have shown, it is sufficiently powerful to explain some surprising asymptotic behaviour for simple standard programs in \mK,
but the applicability of this framework in practice for real implementations of \mK requires further investigation.
Two key aspects that determine practical applicability are the admissibility of the imposed requirements and
the correspondence of specific \mK implementations to the reference operational semantics, which should be studied individually for each application.
We see our work as the ground for the future development of methods for analyzing the cost of interleaving search.

Our approach imposes three requirements on the analyzed programs: disjunctive normal form, uniqueness of answers, and grounding of relational calls.
The first two are rather non-restrictive: DNF is equivalent to the description of relation as a set of Horn clauses in \textsc{Prolog},
and the majority of well-known examples in \mK are written in this or very similar form. Repetition of answers is usually an indication
of a mistake in a program~\cite{WillsThesis}. The groundness condition is more serious: it prohibits program execution from presenting infinitely
many individual ground solutions in one answer using free variables, which is a useful pattern. At the same time, this requirement is
not unique for our work (the framework for \textsc{CASLOG} system mentioned above imposes exactly the same condition) and the experience
shows that many important kinds of programs satisfy it (although it is hard to characterize the class of such programs precisely).
Relaxing any of these restrictions will likely mess up the current relatively compact description of symbolic execution (for
the conditions on relational calls) or the form of the extracted inequalities (for the DNF condition).

Also, for now, we confine ourselves to the problem of estimating the time of the full search for a given goal. Estimating the time before
the first (or some specific) answer is believed to be an important and probably more practical task. Unfortunately, the technique we describe
can not be easily adjusted for this case. The reason for this is that the reasoning about time (scheduling time in particular) in our
terms becomes non-compositional for the case of interrupted search: if an answer is found in some branch, the search is cut short in
other branches, too. Dealing with such a non-compositionality is a subject of future research.






\bibliographystyle{splncs04}
\bibliography{main}

\appendix

\setlength{\abovecaptionskip}{0pt plus 3pt minus 2pt}
\setlength{\belowcaptionskip}{0pt plus 3pt minus 2pt}

\abovedisplayskip1mm
\belowdisplayskip1mm
\abovedisplayshortskip1mm
\belowdisplayshortskip1mm

\setlength{\topsep}{5pt}
\setlength{\partopsep}{5pt plus 0pt minus 0pt}
\setlength{\parskip}{5pt}
\setlength{\parindent}{10pt}

\clearpage
\section{Proofs}
\label{sec:proofs_appendix}

\begin{lemma}
\label{lem:theta_constant}
Let $C_f$ and $C_g$ be real positive constants. If

\[ h\,(x) = f\,(x) + C_f + \Theta\,(g\,(x) + C_g) \]

and value $g\,(x)$ is positive for arbitrary $x$ then

\[ h\,(x) = f\,(x) + \Theta\,(g\,(x)) \]

\end{lemma}
\begin{proof}
\begin{enumerate}
\item For some real postive $C_1$: $h\,(x) \ge f\,(x) + C_f + C_1 \cdot (g\,(x) + C_g)$, thus $h\,(x) \ge f\,(x) + C_1 \cdot g\,(x)$
\item For some real postive $C_2$: $h\,(x) \le f\,(x) + C_f + C_2 \cdot (g\,(x) + C_g)$, thus $h\,(x) \ge f\,(x) + (C_f + C_2 + C_2 \cdot C_g) \cdot g\,(x)$
\qed
\end{enumerate}
\end{proof}

\begin{lemma}
\label{lem:theta_sum}
If

\[ h_1\,(x) = f_1\,(x) + \Theta\,(g_1\,(x)) \]

and

\[ h_2\,(x) = f_2\,(x) + \Theta\,(g_2\,(x)) \]

then

\[ (h_1\,(x) + h_2\,(x)) = (f_1\,(x) + f_2\,(x)) + \Theta\,(g_1\,(x) + g_2\,(x)) \]

\end{lemma}
\begin{proof}
Both lower- and upper-bound constants are the sums of the two corresponding constants for the given approximations.
\qed
\end{proof}

\begin{lemma}
\label{lem:theta_absorb}
Let $g_1\,(x) \ge g_2\,(x) > 0$ for all $x$. If

\[ h\,(x) = f\,(x) + \Theta\,(g_1\,(x) + g_2\,(x)) \]

then

\[ h\,(x) = f\,(x) + \Theta\,(g_1\,(x)) \]

\end{lemma}
\begin{proof}
\begin{enumerate}
\item For some real  postive $C_1$: \[ h\,(x) \ge f\,(x) + C_1 \cdot (g_1\,(x) + g_2\,(x)) \ge f\,(x) + C_1 \cdot g_1\,(x) \]
\item For some real  postive $C_2$: \[ h\,(x) \le f\,(x) + C_2 \cdot (g_1\,(x) + g_2\,(x)) \le f\,(x) + 2 C_2 \cdot g_1\,(x) \]
\qed
\end{enumerate}
\end{proof}

\begin{lemma}
\label{lem:task_measure_equations}
  If

  \[\taskst{g}{e} \xrightarrow{l} s^\prime\]

  then

  \[
    \begin{array}{rcl}
    d\,(\taskst{g}{e}) &=& d\,(s^\prime) + 1\\
    t\,(\taskst{g}{e}) &=& t\,(s^\prime) + 1
    \end{array}
  \]
\end{lemma}
\begin{proof}
    Immediately from the definitions of the estimated values.
    \qed
\end{proof}

\repeatlemma{lem:sum_measure_equations}
\begin{proof}
Induction on the sum of values $d\,(s_1) + d\,(s_2)$. Both equations easily hold for both the rules \ruleno{DisjStop} and \ruleno{DisjStep} if we apply
the inductive hypothesis for the next state.
\qed
\end{proof}

\repeatlemma{lem:times_measure_equations}
\begin{proof}
Induction on the value $d\,(s \times g)$. Both equations easily hold for all the rules \ruleno{ConjStop}, \ruleno{ConjStopAns}, \ruleno{ConjStep}, \ruleno{ConjStepAns}
if we apply the inductive hypothesis for the next state.
\qed
\end{proof}

\repeatlemma{lem:otimes_t_approximation}
\begin{proof}
After unfolding the defintions and throwing out the identical parts we have the following statement:

\[ \begin{array}{l}
d\,(s) + \smashoperator[lr]{\sum\limits_{a_i \in \tra{s}}} \min\,\{2\cdot d\,(\taskst{g}{a_i}) - 1, 2\cdot d\,(s'_i \otimes g)\}  = \\
= \Theta\,(d\,(s) + \smashoperator[lr]{\sum\limits_{a_i \in \tra{s}}} d\,(\taskst{g}{a_i}) - \smashoperator{\maxd\limits_{a_i \in \tra{s}}} d\,(\taskst{g}{a_i})) \\
\end{array} \]

If $\tra{s}$ is empty, the statement is trivial.

If $\tra{s}$ is not empty, $\displaystyle{\maxd}$ turns into a simple $\max$. We establish lower and upper bounds separately.

\begin{enumerate}
  \item
  Lower bound. Let's show that 
  
  \[ \begin{array}{l}
  d\,(s) + \smashoperator[lr]{\sum\limits_{a_i \in \tra{s}}} \min\,\{2\cdot d\,(\taskst{g}{a_i}) - 1, 2\cdot d\,(s'_i \otimes g)\}  \ge \\
  \ge d\,(s) + \smashoperator[lr]{\sum\limits_{a_i \in \tra{s}}} d\,(\taskst{g}{a_i}) - \smashoperator{\max\limits_{a_i \in \tra{s}}} d\,(\taskst{g}{a_i}) \\
  \end{array} \]

  First, we can decrease both arguments of $\min$:

  \[ \begin{array}{l}
  d\,(s) + \smashoperator[lr]{\sum\limits_{a_i \in \tra{s}}} \min\,\{2\cdot d\,(\taskst{g}{a_i}) - 1, 2\cdot d\,(s'_i \otimes g)\}  \ge \\
  \ge d\,(s) + \smashoperator[lr]{\sum\limits_{a_i \in \tra{s}}} \min\,\{ d\,(\taskst{g}{a_i}), d\,(s'_i \otimes g)\} \\
  \end{array} \]

  Let's consider two cases.

  \begin{enumerate}
    \item The minimum in this expression is always reached on the first argument. Then
    
    \[ \begin{array}{l}
        d\,(s) + \smashoperator[lr]{\sum\limits_{a_i \in \tra{s}}} \min\,\{ d\,(\taskst{g}{a_i}), d\,(s'_i \otimes g)\} = \\
        = d\,(s) + \smashoperator[lr]{\sum\limits_{a_i \in \tra{s}}} d\,(\taskst{g}{a_i}) \ge \\
        \ge d\,(s) + \smashoperator[lr]{\sum\limits_{a_i \in \tra{s}}} d\,(\taskst{g}{a_i}) - \smashoperator{\max\limits_{a_i \in \tra{s}}} d\,(\taskst{g}{a_i}) \\
  \end{array} \]
  
    \item $a_m$ is the first answer, such that the minimum for $a_m$ is reached on the second argument. Then the answers up to $a_m$ are sufficient to prove the bound. 
    
    \[ \begin{array}{l}
        d\,(s) + \smashoperator[lr]{\sum\limits_{a_i \in \tra{s}}} \min\,\{ d\,(\taskst{g}{a_i}), d\,(s'_i \otimes g)\} \ge \\
        \ge d\,(s) + \smashoperator[lr]{\sum\limits_{a_i \in \{ a_1, \dots, a_m \}}} \min\,\{ d\,(\taskst{g}{a_i}), d\,(s'_i \otimes g)\} = \\
        = d\,(s) + \smashoperator[lr]{\sum\limits_{a_i \in \{ a_1, \dots, a_{m - 1} \}}} d\,(\taskst{g}{a_i}) + d\,(s'_m \otimes g) = \\
        = d\,(s) + \smashoperator[lr]{\sum\limits_{a_i \in \{ a_1, \dots, a_{m - 1} \}}} d\,(\taskst{g}{a_i}) + d\,(s'_m) + \smashoperator[lr]{\sum\limits_{a_i \in (\tra{s} \setminus \{ a_{1}, \dots, a_{m} \})}} d\,(\taskst{g}{a_i})  \ge \\
        \ge d\,(s) + \smashoperator[lr]{\sum\limits_{a_i \in \tra{s}}} d\,(\taskst{g}{a_i}) -  d\,(\taskst{g}{a_m}) \ge \\
        \ge d\,(s) + \smashoperator[lr]{\sum\limits_{a_i \in \tra{s}}} d\,(\taskst{g}{a_i}) -  \smashoperator{\max\limits_{a_i \in \tra{s}}} d\,(\taskst{g}{a_i}) \\
  \end{array} \]
  \end{enumerate}

  \item 
  Upper bound. Let's show that 
  
  \[ \begin{array}{l}
  d\,(s) + \smashoperator[lr]{\sum\limits_{a_i \in \tra{s}}} \min\,\{2\cdot d\,(\taskst{g}{a_i}) - 1, 2\cdot d\,(s'_i \otimes g)\}  \le \\
  \le 4 \cdot (d\,(s) + \smashoperator[lr]{\sum\limits_{a_i \in \tra{s}}} d\,(\taskst{g}{a_i}) - \smashoperator{\max\limits_{a_i \in \tra{s}}} d\,(\taskst{g}{a_i})) \\
  \end{array} \]
  
  For the upper bound we replace every $\min$ with any of its arguments (the result can only increase). Let $a_m = \argmax{a_i \in \tra{s}} d\,(\taskst{g}{a_i})$. Then for the upper bound let's
  go with the second argument of $\min$ for $a_m$ and with the first argument for all others.
  
   \[ \begin{array}{l}
  d\,(s) + \smashoperator[lr]{\sum\limits_{a_i \in \tra{s}}} \min\,\{2\cdot d\,(\taskst{g}{a_i}) - 1, 2\cdot d\,(s'_i \otimes g)\}  \le \\
  \le d\,(s) + 2 \cdot d\,(s'_m \otimes g) + \smashoperator[lr]{\sum\limits_{a_i \in (\tra{s} \setminus \{a_m\})}} (2\cdot d\,(\taskst{g}{a_i})) \le \\
  \le d\,(s) + 2 \cdot d\,(s'_m) + 2\cdot\smashoperator[lr]{\sum\limits_{a_i \in (\begin{array}{l} \tra{s} \setminus \\ \{a_1, \dots, a_m\}) \end{array}}} d\,(\taskst{g}{a_i}) + 2\cdot\smashoperator[lr]{\sum\limits_{a_i \in (\tra{s} \setminus \{a_m\})}} d\,(\taskst{g}{a_i}) \le \\
  \le d\,(s) + 2 \cdot d\,(s) + 2\cdot\smashoperator[lr]{\sum\limits_{a_i \in (\tra{s} \setminus \{a_m\})}} d\,(\taskst{g}{a_i}) + 2\cdot\smashoperator[lr]{\sum\limits_{a_i \in (\tra{s} \setminus \{a_m\})}} d\,(\taskst{g}{a_i}) \le \\
  \le 4 \cdot d\,(s) + 4\cdot\smashoperator[lr]{\sum\limits_{a_i \in (\tra{s} \setminus \{a_m\})}} d\,(\taskst{g}{a_i}) \\
  \end{array} \] $\qed$

\end{enumerate}

\end{proof}

\begin{lemma}
\label{lem:times_gen_measure_approximations}

Let $s = ((s_0 \otimes g_1) \dots \otimes g_k)$ and let $A_i$ be a set of all answers that are passed to $g_i$, i.e.

\[
\begin{array}{rcl}
A_1 &=& \tra{s_0} \\
A_{i + 1} & = & \bigcup\limits_{a \in A_i} \tra{\taskst{g_i}{a}} 
\end{array}
\]

Then

\[
\begin{array}{rcrl}
d\,(s) &=& & d\,(s_0) + \sum\limits_{1 \le i \le k} \displaystyle\sum\limits_{a \in A_i} d\,(\taskst{g_i}{a}) \\
\\
t\,(s) &=& & t\,(s_0) + \sum\limits_{1 \le i \le k} \displaystyle\sum\limits_{a \in A_i} t\,(\taskst{g_i}{a}) + \\
& & + \Theta\,(& d\,(s_0) + \sum\limits_{1 \le i \le k} \displaystyle\sum\limits_{a \in A_i} d\,(\taskst{g_i}{a}) - \maxd\limits_{a \in A_k}  d\,(\taskst{g_k}{a})) \\
\end{array}
\]

\end{lemma}
\begin{proof}
First we show that for all $i$, $A_{i + 1} = \tra{((s_0 \otimes g_1) \dots \otimes g_i)}$ (it's a simple induction on $i$ and then on the trace).

Then we can prove the statement by induction on $k$.
We unfold $d\,(s \otimes g_{k+1})$ and $t\,(s \otimes g_{k+1})$ using equations in \lemmaword~\ref{lem:times_measure_equations}.
Then we rewrite $d\,(s)$ and $t\,(s)$ in these equations with inductive hypothesis.
For $d\,(s \otimes g_{k+1})$ we get exactly the equation we need.
For $t\,(s \otimes g_{k+1})$ we have a sum of the following parts (first two from the unfolding of $t\,(s)$, last two from the rest of the equation for $t\,(s \otimes g_{k+1})$):
\begin{enumerate}
\item $t\,(s_0) + \displaystyle{\sum\limits_{1 \le i \le k}} \displaystyle\sum\limits_{a \in A_i} t\,(\taskst{g_i}{a})$
\item $ \Theta\,(d\,(s_0) +  \displaystyle{\sum\limits_{1 \le i \le k}} \displaystyle\sum\limits_{a \in A_i} d\,(\taskst{g_i}{a}) - \maxd\limits_{a \in A_k}  d\,(\taskst{g_k}{a}))$
\item $\displaystyle\sum\limits_{a \in A_{k+1}} t\,(\taskst{g_{k+1}}{a})$
\item $ \Theta\,(d\,(s_0) +  \displaystyle{\sum\limits_{1 \le i \le k}} \displaystyle\sum\limits_{a \in A_i} d\,(\taskst{g_i}{a}) + \displaystyle\sum\limits_{a \in A_k} d\,(\taskst{g_{k+1}}{a}) - \maxd\limits_{a \in A_{k + 1}}  d\,(\taskst{g_k}{a}))$
\end{enumerate}

We can see that the second part is subsumed by the last part (by \lemmaword~\ref{lem:theta_absorb}). The rest gives exactly the equation we need.
\qed
\end{proof}

Here is the general definition of well-formedness of states from~\cite{CertifiedSemantics}.

\begin{definition}
  Well-formedness condition for extended states:

  \begin{itemize}
  \item $\diamond$ is well-formed;
  \item $\inbr{g, \sigma, n}$ is well-formed iff $\fv{g}\cup\Dom\,(\sigma)\cup\VRan\,(\sigma)\subseteq\{\alpha_1,\dots,\alpha_n\}$;
  \item $s_1\oplus s_2$ is well-formed iff $s_1$ and $s_2$ are well-formed;
  \item $s\otimes g$ is well-formed iff $s$ is well-formed and for all leaf triplets $\inbr{\_,\_,n}$ in $s$ it is true that $\fv{g}\subseteq\{\alpha_1,\dots,\alpha_n\}$.
  \end{itemize}

\end{definition}

We will need \lemmaword~\ref{lem:measures_changing_env} in the following generalized form.

\begin{lemma}
\label{lem:gen_measures_changing_env}
Let $\pi \colon \{ \alpha_1, \dots, \alpha_N \} \to \{ \alpha_1, \dots, \alpha_{N'} \}$ be an injective function on variables and let $R_{\pi}$ be the following inductively defined relation on states:
\[ \begin{array}{lcl}
\Diamond R_{\pi} \Diamond & & \\
\taskst{g}{\mkenv{\sigma}{n}} R_{\pi} \taskst{g'}{\mkenv{\sigma'}{n'}} & \textit{iff} & g \sigma \pi = g' \sigma' \\
(s_1 \oplus s_2) R_{\pi} (s'_1 \oplus s'_2) & \textit{iff} & s_1 R_{\pi} s'_1 \land s_2 R_{\pi} s'_2  \\
(s \otimes g) R_{\pi} (s' \otimes g') & \textit{iff} & s R_{\pi} s' \land \\
\multicolumn{3}{l}{\textit{\quad for all substates $\taskst{g_i}{\mkenv{\sigma_i}{n_i}}$ in $s$}}\\
\multicolumn{3}{l}{\textit{\quad and corresponding substates $\taskst{g'_i}{\mkenv{\sigma'_i}{n'_i}}$ in $s'$,}} \\
\multicolumn{3}{l}{\quad g \sigma_i \pi = g' \sigma'_i} \\
\end{array} \]

Then for any two well-formed states $s$ and $s'$, such that all counters occuring in $s$ are less or equal than some $n$ and all counters occuring in $s'$ are less or equal than some $n'$ and $s R_{\pi} s'$ for some injective function $\pi \colon \{ \alpha_1, \dots, \alpha_n \} \to \{ \alpha_1, \dots, \alpha_n' \}$,

\[ d\,(s) = d\,(s') \]
and

\[t\,(s) = t\,(s') \]

and there is a bijection $b$ between sets of answers $\tra{s}$ and $\tra{s'}$ such that for any answer $a = \mkenv{\sigma_r}{n_r} \in \tra{s}$ there is a corresponding answer $b(a) = \mkenv{\sigma'_r}{n'_r} \in \tra{s'}$, s.t. $\sigma_r = \sigma \delta$ for some $\sigma$ that is a subtitution in some leaf substate of $s$ and $\sigma'_r = \sigma' \delta'$ for $\sigma'$ that is the substitutution of the corresponding leaf substate of $s'$ and there is an injective function $\pi_r \colon \{ \alpha_1, \dots, \alpha_{n_r} \} \to \{ \alpha_1, \dots, \alpha_{n'_r} \}$ such that $\pi_r \succ \pi$ and $\pi \delta' = \delta \pi_r$.
\end{lemma}
\begin{proof}
We prove it by induction on the length of the trace for $s$; simultaneously we prove that the next states in the traces for $s$ and $s'$ also satisfy the relation $R_{\pi_r}$ for some $\pi_r$ (s.t. $\pi_r \succ \pi$). The equalities $d\,(s) = d\,(s')$ and $t\,(s) = t\,(s')$ are obvious in this induction, because states in the relation $R_{\pi}$ always have the same form (and therefore the same left height). To prove the fact about the bijection between answers and the fact that the relation holds for the next states we conduct an internal induction on the relation of operational semantics step. When we move through the introduction of the fresh variable we extend the injective function changing the variable by a binding between new fresh variables. In the base case of unification, when we extend the substitutions by the most general unifiers, we have the fact about the bijection between the sets of answers (singleton in this case) for the same injective renaming function by definition of the unification algorithm (we may change the names of variables before the unification and the result will be the same as if we do it after the unification):

\[ \pi \, mgu\,(t_1 \sigma \pi, t_2 \sigma \pi) = mgu\,(t_1 \sigma, t_2 \sigma) \, \pi  \]

For the case when we incorporate the answer obtained at this step in the next state (in rules $\ruleno{ConjStopAns}$ and $\ruleno{ConjStepAns}$) we use the statement about the bijection between the sets of answers from the inductive hypothesis to prove that the next states satisfy the relation:

\[ g \sigma \delta \pi_r = g \sigma \pi \delta' = g' \sigma' \delta'  \] 

All other cases naturally follow from inductive hypotheses.\qed

\end{proof}

\repeatlemma{lem:measures_changing_env}
\begin{proof}
It is a special case of \lemmaword~\ref{lem:gen_measures_changing_env}.\qed
\end{proof}

\begin{lemma}
\label{lem:update_substitutions_FV}

Let $\taskst{g_0}{\mkenv{\sigma_0}{n_0}}$ be a leaf state. Then for every answer $\mkenv{\sigma'}{n'}$ in $\tra{\taskst{g_0}{\mkenv{\sigma_0}{n_0}}}$, $\sigma' = \sigma_0 \delta$ for some substitution $\delta$, such that $\Dom\,(\delta) \cup \VRan\,(\delta) \subset FV\,(g_0 \sigma_0) \cup \{ \alpha_i \mid i > n_0 \}$.

\end{lemma}
\begin{proof}

First, we need some notions to generalize the statement.

Let $ENV\,(s)$ be the set of environments that occur in the given state.

\[ \begin{array}{lcl}
ENV\,(\Diamond) & = & \emptyset \\
ENV\,(\taskst{g}{e}) & = & \{ e \} \\
ENV\,(s_1 \oplus s_2) & = & ENV\,(s_1) \cup ENV\,(s_2) \\
ENV\,(s \otimes g) & = & ENV\,(s) 
\end{array} \]

Now, let's generalize the set of variables updated by answers from the statement to an arbitrary state.

\[ \begin{array}{lcl}
\Delta\,(\Diamond) & = & \emptyset \\
\Delta\,(\taskst{g}{\mkenv{\sigma}{n}}) & = & FV\,(g \sigma) \cup \{ \alpha_i \mid i > n \} \\
\Delta\,(s_1 \oplus s_2) & = & \Delta\,(s_1) \cup \Delta\,(s_2) \\
\Delta\,(s \otimes g) & = & \Delta\,(s_1) \cup \bigcup\limits_{\mkenv{\sigma}{n} \in ENV\,(s)} FV\,(g \sigma)
\end{array} \]

Now we can generalize the statement but for one semantical step only: if $s \xrightarrow{l} s'$, then the following three conditions hold:

\begin{enumerate}
\item $\Delta\,(s) \supset \Delta\,(s')$
\item If $l = \mkenv{\sigma'}{n'}$ then there exists $\mkenv{\sigma}{n} \in ENV\,(s)$ and substitution $\delta$, such that $\sigma' = \sigma \delta$ and $n' = n$ and $\Dom\,(\delta) \cup \VRan\,(\delta) \subset \Delta\,(s)$
\item For any $\mkenv{\sigma'}{n'} \in ENV\,(s')$ there exists $\mkenv{\sigma}{n} \in ENV\,(s)$ and substitution $\delta$, such that $\sigma' = \sigma \delta$ and $n' \ge n$ and $\Dom\,(\delta) \cup \VRan\,(\delta) \subset \Delta\,(s)$
\end{enumerate}

We prove it by the induction on semantical step relation (we have to prove all three conditions simultaneously).

\begin{enumerate}
\item The first condition is simple for the steps from leaf goals: the counters of occupied variables can only increase and the sets of free variables of subgoals can only decrease (except for the case of fresh variable introduction, where a new free variable appears, but it is greater than the counter); in case of invocation we use the fact that the body of any relation is closed (there are no free variables except for the arguments). For the $\ruleno{ConjStopAns}$ rule we use the second condition: the next step has the environment that updates one of the environments in $s$, so it does not introduce new variables in $\Delta$ and the counter also may only increase. For the $\ruleno{ConjStep}$ rule we use the third condition: all substitutions from environments of updated state after application to a goal do not introduce new variables in $\Delta$, the rest is handled by the inductive hypothesis. For the $\ruleno{ConjStepAns}$ rule we combine two previous arguments and for other cases the first condition is obvious from the inductive hypothesis.

\item The second condition needs to be proven only for the $\ruleno{UnifySuccess}$ rule (where it follows from the properties of the unification algorithm) and for the rules $\ruleno{DisjStop}$ and $\ruleno{DisjStep}$ (where it is obvious from the inductive hypothesis).

\item The third condition follows simply in all cases from the inductive hypothesis and from the second condition (for the rules $\ruleno{ConjStep}$ and $\ruleno{ConjStepAns}$ where the answer is incorporated in the next state).
\end{enumerate}

Now, the statement of the lemma follows from the generalized statement for one step: at each step substitutions in the answers and in the next step are composed with some additional substitutions that manipulate with only variables from the set $\Delta$ for this step, which is a subset of the set $\Delta$ in the beginning, which is exactly what we need.
\qed

\end{proof}

\repeatlemma{lem:symbolic_unification_soundness}
\begin{proof} $ $

From the corectness of the Robinson's unification algoithm we know that a substitution unifies the pair of terms $(t_1, t_2)$ iff it unifies all pairs of terms from the set $\{ (x, \delta\,(x)) \mid x \in \Dom\,(\delta) \}$ (because we obtain $\delta$ from the pair $(t_1, t_2)$ by Robinson's algorithm that maintains equivalent unification problem).

First, let's notice that similarly a substitution unifies the pair of terms $(t_1 \rho, t_2 \rho)$ iff it unifies all pairs of terms  $T = \{ (\rho\,(x), \delta\,(x) \rho) \mid x \in \Dom\,(\delta) \}$: $\nu$ is such substitution iff $\rho \nu$ unifies the terms $(t_1, t_2)$. And then also any most general unifier for $(t_1 \rho, t_2 \rho)$ is the most general unifier for $T$ and vice versa (by definition).

So now we need to show that $T$ is unifiable iff the unique $\rho'$ from the statement of the lemma exists (and that the most general unifier for $T$ can be defined with $\rho'$ and $\delta$). In both directions we will use the induction on the construction of the set of variables $U$, so lets consider the following sequence $U_i$: $U_0 = V$ and $U_{i+1} = \{ U_i \cup \bigcup\limits_{x \in U_{i}} FV\,(\delta\,(x)) \}$ (so $U$ is $U_l$ such that $U_l = U_{l + 1}$).

\begin{enumerate}
\item Suppose there is a substitution $\tau$ that unifies all the terms in $T$. Let's show that there is a unique $\rho'$ such that $\rho' \succ \rho$ and $\forall (y, \, t) \in \constr{\delta}{U}, \\ \rho'(y) = t \rho'$.

We know that $\rho\,(x) \tau = \delta\,(x) \rho \tau$ for all $x \in \Dom\,(\delta)$. We need the same condition for $\rho'$ for all $x \in U \cap \Dom\,(\delta)$. We can now show by induction on $i$ that for all variables $x \in U_i$ ($i \ge 1$) the value $\tau\,(x)$ is ground and uniquely defined for a given $\rho$, so they can be taken as values of $\rho'$ on variables from $U \setminus V$ and they are the only possible values. First, look at a pair $(\rho\,(x), \delta\,(x) \rho)$ in $T$ for some $x \in V$. We know that $\rho\,(x) \tau = \delta\,(x) \rho \tau$ and the term on the lhs is ground and uniquely defined by $\rho$. So the values of $\tau$ for all free variables of $\delta\,(x) \rho$ are ground and uniquely defined by $\rho$, too. If we do it for all such pairs in $T$ we will get the statement for $U_1$, then we can repeat this reasoning by induction for all $U_i$.

\item Now suppose there is $\rho'$ such that $\rho' \succ \rho$ and $\forall (y, \, t) \in \constr{\delta}{U}, \rho'(y) = t \rho'$. Let's construct the most general unifier for $T$ using the Robinson's algorithm.

Let's split $T$ on $T_U = \{ (\rho\,(x), \delta\,(x) \rho) \mid x \in U \cap \Dom\,(\delta) \}$ and $T_{-U} = \{ (\rho\,(x), \delta\,(x) \rho) \mid x \in  \Dom\,(\delta) \setminus U \}$. We will be applying rules from the Robinson's algorithm to $T$ for pairs of terms from $T_U$.

First, let's look at some pair $(\rho\,(x), \delta\,(x) \rho)$ for $x \in V \cap \Dom\,(\delta)$. By definition of $\rho'$ we have $\rho\,(x) \rho' = \delta\,(x) \rho \rho'$. The first term in the pair is ground, the second one may contain free variables (then they are variables from $U_1$). If the second term is ground, too, they are equal and we can delete this pair. Otherwise, using decomposition rule we decompose this pair to pairs of terms with second term being variable. After this, we will have pairs $(\rho'(y), y)$ for all $y \in FV\,(\delta\,(x) \rho)$. After that we do it for all such pairs for all variables from $V \cap \Dom\,(\delta)$, this pairs will turn into swapped bindings $(\rho'(y), y)$ for all $y \in U_1$ (maybe with repetitions). We then can discard the duplicates, swap the elements and apply this bindings in the rest of $T$. Now all the pairs $(\rho\,(x), \delta\,(x) \rho)$ for $x \in (U_1 \setminus V) \cap \Dom\,(\delta)$ after substitution have ground terms as the left term and we can repeat the transformation for all these pairs. We can repeat this process by induction on $i$ for $x \in U_i$, until $U_i$ becomes equal to $U$.

After this application of rules all pairs from $T_U$ are decomposed and turned into the substitution (as a set of bindings) $\rho'\restriction_{U \setminus V}$. On the other hand the pairs from $T_{-U}$  are not decomposed, just applied substitutions to them so every pair from this part still has some variable $z$ as the first term (because $z\not\in U$) and term $\delta\,(z) \rho (\rho'\restriction_{U \setminus V})$ as the second term. So we can see that we turned the set of terms $T$ into the substitution ${(\delta\restriction_{\Dom\,(\delta) \setminus U} \rho')\restriction_{\Dom\,(\delta) \setminus V}}$ which equals ${(\delta \rho')\restriction_{\Dom\,(\delta) \setminus V}}$ because $\rho'$ unifies bindings in $\delta$, so this substitution is the most general unifier for $T$ and therefore for terms $t_1 \rho$ and $t_2 \rho$. Then this substitution is alpha-equivalent to $mgu\,(t_1 \rho, t_2 \rho)$ (because most general unifiers are unique up to alpha-equivalence). So if we take $\rho mgu\,(t_1 \rho, t_2 \rho)$, we will get substitution alpha-equivalent to the substitution ${\rho ((\delta \rho')\restriction_{\Dom\,(\delta) \setminus V})}$ which equals to the substitution $\delta \rho'$ (obviously separately for variables from $V$ andfrom outside $V$).
\qed

\end{enumerate}
\end{proof}

\repeattheorem{extracted_approximations}
\begin{proof}
$ $\newline
\begin{enumerate}
\item Suppose we have proven this statement for $g \in C_{nf}$. Let's show it holds for $g \in D_{nf}$.
	\begin{enumerate}
	\item First, we prove it for  $g \in F_{nf}$ by induction on the goal.
	After unfolding each fresh constructor we get a goal with the same scheme.	
	Also going through each fresh constuct increases the values of $d\,(\cdot)$ and $t\,(\cdot)$ by $1$ by \lemmaword~\ref{lem:task_measure_equations}, this additional constant can be deleted by \lemmaword~\ref{lem:theta_constant}.
	
	\item Now we prove it for $g \in D_{nf}$ by induction on the goal. Let $g = \disjgoal{g_1}{g_2}$. Let $\schemewithvset{\mathfrak{S}_1}{V}$ be $\schemewithvset{\mathfrak{S}_2}{V}$ be children of the root in $\schemewithvset{\mathfrak{S}}{V}$.
		
	By \lemmaword~\ref{lem:task_measure_equations} and \lemmaword~\ref{lem:sum_measure_equations} and \lemmaword~\ref{lem:measures_changing_env} we have the following equations:

    \[ \begin{array}{lcl}
	d\,(\taskst{\disjgoal{g_1}{g_2}}{e_{init}}) &=& d\,(\taskst{g_1}{e_{init}}) + d\,(\taskst{g_2}{e_{init}}) + 1 \\
	\\
	t\,(\taskst{\disjgoal{g_1}{g_2}}{e_{init}}) &=& t\,(\taskst{g_1}{e_{init}}) + t\,(\taskst{g_2}{e_{init}}) \\
	&& + \min\,(2 d\,(\taskst{g_1}{e_{init}}) - 1, 2 d\,(\taskst{g_2}{e_{init}})) + 1
	\end{array} \]
	
	After rewriting the right part with inductive hypotheses (combining them using \lemmaword~\ref{lem:theta_constant} and \lemmaword~\ref{lem:theta_sum}) we get the following approximations.
	
	 \[ \begin{array}{lcl}
	d\,(\taskst{\disjgoal{g_1 \rho}{g_2 \rho}}{e_{init}}) &=& \mathcal{D}(\schemewithvset{\mathfrak{S}_1}{V})(\rho) + \mathcal{D}(\schemewithvset{\mathfrak{S}_2}{V})(\rho) + \Theta\,(1) \\
	\\
	t\,(\taskst{\disjgoal{g_1 \rho}{g_2 \rho}}{e_{init}}) &=& \mathcal{T}(\schemewithvset{\mathfrak{S}_1}{V})(\rho) + \mathcal{T}(\schemewithvset{\mathfrak{S}_2}{V})(\rho) + \\
	& & \Theta\,(\min\,(\mathcal{D}(\schemewithvset{\mathfrak{S}_1}{V})(\rho), \mathcal{D}(\schemewithvset{\mathfrak{S}_2}{V})(\rho)) + \\
	& & + (\mathcal{D}(\schemewithvset{\mathfrak{S}_1}{V})(\rho) - \smashoperator{\maxd\limits_{\taskst{g_i}{e_i} \in \mathcal{L}(\schemewithvset{\mathfrak{S}_1}{V})(\rho)}} d\,(\taskst{g_i}{e_i})) + \\
	& & + (\mathcal{D}(\schemewithvset{\mathfrak{S}_2}{V})(\rho) - \smashoperator{\maxd\limits_{\taskst{g_i}{e_i} \in \mathcal{L}(\schemewithvset{\mathfrak{S}_2}{V})(\rho)}} d\,(\taskst{g_i}{e_i})) + 1) \\
	\end{array} \]
	
	For $d\,(\cdot)$ it is exactly what we need. 
   
   	W.l.o.g. let's suppose the minimum in the approximation for $t\,(\cdot)$ is achieved at the first argument.
   
   	Let's consider two cases: which of $\smashoperator{\maxd\limits_{\taskst{g_i}{e_i}  \in \mathcal{L}(\schemewithvset{\mathfrak{S}_l}{V})(\rho)}}d\,(\taskst{g_i}{e_i})$ is the maximal leaf for the whole scheme.
   
   		\begin{enumerate}
   		\item Suppose $\smashoperator{\maxd\limits_{\taskst{g_i}{e_i} \in \mathcal{L}(\schemewithvset{\mathfrak{S}_1}{V})(\rho)}} d\,(\taskst{g_i}{e_i}) \le \smashoperator{\maxd\limits_{\taskst{g_i}{e_i} \in \mathcal{L}(\schemewithvset{\mathfrak{S}_2}{V})(\rho)}} d\,(\taskst{g_i}{e_i})$.
   		
   		Then we can absorb the summand $(\mathcal{D}(\schemewithvset{\mathfrak{S}_1}{V})(\rho) - \smashoperator{\maxd\limits_{\taskst{g_i}{e_i} \in \mathcal{L}(\schemewithvset{\mathfrak{S}_1}{V})(\rho)}} d\,(\taskst{g_i}{e_i}))$ under $\Theta$ by the larger summand $\mathcal{D}(\schemewithvset{\mathfrak{S}_1}{V})(\rho)$ (which came from $\min$) by \lemmaword~\ref{lem:theta_absorb}. We get the following approximation which is exactly what we need:
   		
   		\[ \begin{array}{lcl}
		t\,(\taskst{\disjgoal{g_1 \rho}{g_2 \rho}}{e_{init}}) &=& \mathcal{T}(\schemewithvset{\mathfrak{S}_1}{V})(\rho) + \mathcal{T}(\schemewithvset{\mathfrak{S}_2}{V})(\rho) + \\
		& & \Theta\,(\mathcal{D}(\schemewithvset{\mathfrak{S}_1}{V})(\rho) + \\
		& & + (\mathcal{D}(\schemewithvset{\mathfrak{S}_2}{V})(\rho) - \smashoperator{\maxd\limits_{\taskst{g_i}{e_i} \in \mathcal{L}(\schemewithvset{\mathfrak{S}_2}{V})(\rho)}} d\,(\taskst{g_i}{e_i})) + 1) \\
		\end{array} \]
		
		\item Suppose $\smashoperator{\maxd\limits_{\taskst{g_i}{e_i} \in \mathcal{L}(\schemewithvset{\mathfrak{S}_1}{V})(\rho)}} d\,(\taskst{g_i}{e_i}) > \smashoperator{\maxd\limits_{\taskst{g_i}{e_i} \in \mathcal{L}(\schemewithvset{\mathfrak{S}_2}{V})(\rho)}} d\,(\taskst{g_i}{e_i})$.
		
		Then we establish the lower and the upper bounds separately.
		
			\begin{enumerate}
   			\item For the lower bound we again first absorb the summand $(\mathcal{D}(\schemewithvset{\mathfrak{S}_1}{V})(\rho) - \smashoperator{\maxd\limits_{\taskst{g_i}{e_i} \in \mathcal{L}(\schemewithvset{\mathfrak{S}_1}{V})(\rho)}} d\,(\taskst{g_i}{e_i}))$ by \lemmaword~\ref{lem:theta_absorb} to get the following.

			\[ \begin{array}{lcl}
			t\,(\taskst{\disjgoal{g_1 \rho}{g_2 \rho}}{e_{init}}) &\ge& \mathcal{T}(\schemewithvset{\mathfrak{S}_1}{V})(\rho) + \mathcal{T}(\schemewithvset{\mathfrak{S}_2}{V})(\rho) + \\
			& & C_1 \cdot (\mathcal{D}(\schemewithvset{\mathfrak{S}_1}{V})(\rho) + \\
			& & + (\mathcal{D}(\schemewithvset{\mathfrak{S}_2}{V})(\rho) - \smashoperator{\maxd\limits_{\taskst{g_i}{e_i} \in \mathcal{L}(\schemewithvset{\mathfrak{S}_2}{V})(\rho)}} d\,(\taskst{g_i}{e_i})) + 1) \\
			\end{array} \]		
		
			Then replace $(-\smashoperator{\maxd\limits_{\taskst{g_i}{e_i} \in \mathcal{L}(\schemewithvset{\mathfrak{S}_2}{V})(\rho)}} d\,(\taskst{g_i}{e_i}))$ \\ by $(-\smashoperator{\maxd\limits_{\taskst{g_i}{e_i} \in \mathcal{L}(\schemewithvset{\mathfrak{S}_1}{V})(\rho)}} d\,(\taskst{g_i}{e_i}))$ which is smaller by assumption.
		
			\item For the upper bound we first replace $\mathcal{D}(\schemewithvset{\mathfrak{S}_1}{V})(\rho)$ that came form $\min$ by $\mathcal{D}(\schemewithvset{\mathfrak{S}_2}{V})(\rho)$ which is larger by the assumption and then absorb the summand $(\mathcal{D}(\schemewithvset{\mathfrak{S}_2}{V})(\rho) - \smashoperator{\maxd\limits_{\taskst{g_i}{e_i} \in \mathcal{L}(\schemewithvset{\mathfrak{S}_2}{V})(\rho)}} d\,(\taskst{g_i}{e_i}))$ by it.
			
			\end{enumerate}		
   		\end{enumerate}
	\end{enumerate}

\item Now let's prove the statement of the theorem for $g \in C_{nf}$.

Let $g = { \conjgoal{(\conjgoal{(\conjgoal{g_1}{g_2})}{\dots})}{g_k} }$ with $g_i \in B_{nf}$.

First, notice that the state $\taskst{g}{e_{init}}$ is transformed into the state \\ ${ ((\taskst{g_1}{\mkenv{\varepsilon}{n_{init}(g)}} \otimes g_2) \otimes \dots) \otimes g_k }$ after $(k - 1)$ steps of turning conjunctions into $\otimes$-states. All states during this steps except the resulting one add $(k - 1)$ to the value $d\,(\taskst{g}{e_{init}})$ and $\dfrac{(k - 1)(k - 2)}{2}$ to the value $t\,(\taskst{g}{e_{init}})$, we can hide this constants under $\Theta$ by \lemmaword~\ref{lem:theta_absorb}. \lemmaword~\ref{lem:times_gen_measure_approximations} gives us the approximations of the measures for the state ${ ((\taskst{g_1}{\mkenv{\varepsilon}{n_{init}(g)}} \otimes g_2) \otimes \dots) \otimes g_k }$. To put it in a convenient form we will use the following definitions.

\[ \begin{array}{lcl}
D\,(s, \epsilon) & = & d\,(s) \\
D\,(s, g : \Gamma) & = & d\,(s) + \smashoperator{\sum\limits_{a \in \tra{s}}} D\,(\taskst{g}{a}, \Gamma) \\
\\
T\,(s, \epsilon) & = & t\,(s) \\
T\,(s, g : \Gamma) & = & t\,(s) + \smashoperator{\sum\limits_{a \in \tra{s}}} T\,(\taskst{g}{a}, \Gamma) \\
\\
L\,(s, \epsilon) & = & \{ s \} \\
L\,(s, g : \Gamma) & = & \smashoperator{\bigcup\limits_{a \in \tra{s}}} T\,(\taskst{g}{a}, \Gamma) \\
\end{array}
\]

Now,  \lemmaword~\ref{lem:times_gen_measure_approximations} gives us the following approximations if we denote ${ g_2 \rho : \dots g_k \rho}$ by $\Gamma$.

\[ \begin{array}{lcl}
d\,(\taskst{g \rho}{e_{init}}) &=& D\,(\taskst{g_1 \rho}{\mkenv{\varepsilon}{n_{init}(g)}}, \Gamma) + \Theta\,(1) \\
\\
t\,(\taskst{g \rho}{e_{init}}) &=& T\,(\taskst{g_1 \rho}{\mkenv{\varepsilon}{n_{init}(g)}}, \Gamma) + \\ 
& & \Theta\,(D\,(\taskst{g_1 \rho}{\mkenv{\varepsilon}{n_{init}(g)}}, \Gamma) - \smashoperator{\maxd\limits_{s \in L\,(\taskst{g_1 \rho}{\mkenv{\varepsilon}{n_{init}(g)}}, \Gamma)}} d\,(s) + 1) \\
\end{array} \]

It's the approximation in the form required in the statement of the theorem. What remains to be proven are the following equalities:

\[ \begin{array}{lcl}
D\,(\taskst{g_1 \rho}{\mkenv{\varepsilon}{n_{init}(g)}}, \Gamma) &=& \mathcal{D}(\schemewithvset{\mathfrak{S}}{V})(\rho) \\
T\,(\taskst{g_1 \rho}{\mkenv{\varepsilon}{n_{init}(g)}}, \Gamma) &=& \mathcal{T}(\schemewithvset{\mathfrak{S}}{V})(\rho) \\
\{ d\,(s) \mid s \in L\,(\taskst{g_1 \rho}{\mkenv{\varepsilon}{n_{init}(g)}}, \Gamma)\} &=& \{ d\,(s) \mid s \in \mathcal{L}(\schemewithvset{\mathfrak{S}}{V})(\rho) \} \\
\end{array} \]
	
We can prove them by induction, but first we need to generalize the statement. Let $g$ be a goal from $B_{nf}$ and $\Gamma = { g_1 : \dots : g_m : \epsilon }$~--- a sequence of goals from $B_{nf}$, $\sigma$ be a substitution, $n$ be a fresh variable counter such that the state ${ (((\taskst{g}{\mkenv{\sigma}{n}} \otimes g_1) \otimes \dots) \otimes g_m) }$ is well-formed, and $V$ be a subset of variables $\{ \alpha_1, \dots, \alpha_n \}$. Then if \[ \schemetrans{g}{\Gamma}{\sigma}{n}{V}{\schemewithvset{\mathfrak{S}}{V}} \] then the following equalities hold for any $\rho \colon V \to \mathcal{T}_{\emptyset}$.

\[ \begin{array}{lcl}
D\,(\taskst{g}{\mkenv{\sigma \rho}{n}}, \Gamma) &=& \mathcal{D}(\schemewithvset{\mathfrak{S}}{V})(\rho) \\
T\,(\taskst{g}{\mkenv{\sigma \rho}{n}}, \Gamma) &=& \mathcal{T}(\schemewithvset{\mathfrak{S}}{V})(\rho) \\
\{ d\,(s) \mid s \in L\,(\taskst{g}{\mkenv{\sigma \rho}{n}}, \Gamma)\} &=& \{ d\,(s) \mid s \in \mathcal{L}(\schemewithvset{\mathfrak{S}}{V})(\rho) \} \\
\end{array} \]

(to get from this generalization to the equations above we should take $\sigma = \varepsilon$ and apply \lemmaword~\ref{lem:measures_changing_env} to move $\rho$ from environment to the goal).

We prove the generalized statement by the induction on $\Gamma$ and considering cases when $g$ is an equality or a relational call. The reasoning is exactly the same for all three notions $D\,(\cdot)$, $T\,(\cdot)$ and $L\,(\cdot)$, so we demonstrate only the proof of the equality between $D\,(\cdot)$ and $\mathcal{D}(\cdot)(\cdot)$.

	\begin{enumerate}

	\item Let $\Gamma = \epsilon$ and $g = (\unigoal{t_1}{t_2})$.
	
	In this case we have \[ d\,(\taskst{\unigoal{t_1}{t_2}}{\mkenv{\sigma \rho}{n}}) = d\,(\taskst{\unigoal{t_1 \sigma \rho}{t_2 \sigma \rho}}{e_{init}}) = 1 \]
	
	\item Let $\Gamma = \epsilon$ and $g = (\invokegoal{R^k}{t_1}{t_k})$.
	
	In this case we have \[ d\,(\taskst{\invokegoal{R^k}{t_1}{t_k}}{\mkenv{\sigma \rho}{n}}) = d\,(\taskst{\invokegoal{R^k}{t_1 \sigma \rho}{t_k \sigma \rho}}{e_{init}}) \]
	
	Which is true by \lemmaword~\ref{lem:measures_changing_env}.
	
	\item Let $\Gamma = g' : \Gamma'$ and $g = (\unigoal{t_1}{t_2})$.
	
	\begin{enumerate}
	
	    \item If the terms $t_1 \sigma$ and $t_2 \sigma$ are non-unifiable, we have \[ d\,(\taskst{\unigoal{t_1}{t_2}}{\mkenv{\sigma \rho}{n}}) + \smashoperator{\sum\limits_{mgu\,(t_1 \sigma \rho, t_2 \sigma \rho) = \delta'}} D\,(\taskst{g'}{\mkenv{\sigma \rho \delta'}{n}}, \Gamma') = 1 \]
	    
	    And it is obviously true because the sum is empty since the more specific terms there are non-unifiable also.
	
	    \item If they are unifiable, we have \[ d\,(\taskst{\unigoal{t_1}{t_2}}{\mkenv{\sigma \rho}{n}}) + \smashoperator{\sum\limits_{mgu\,(t_1 \sigma \rho, t_2 \sigma \rho) = \delta'}} D\,(\taskst{g'}{\mkenv{\sigma \rho \delta'}{n}}, \Gamma') = \]
	
	\[ = 1 +
      \smashoperator{\sum\limits_{\substack{ \rho' \colon U \to \grterms \\
                                      \rho' \succ \rho \\
                                      \forall (y, t) \in Cs, \rho'(y) = t \rho'  }}}
           \mathcal{D}(\schemewithvset{\mathfrak{S}}{U})(\rho')  \]
           
    		where
    
    \[ \begin{array}{lcl}
    \delta & = & mgu\,(t_1 \sigma, t_2 \sigma) \\
    U & = & \upd{V}{\delta} \\
    Cs & = & \constr{\delta}{U} \\
	\end{array} \]
	
			The left summands are obviously equal. The rest is basically covered by \lemmaword~\ref{lem:symbolic_unification_soundness}. By this lemma there exists a most general unifier $\delta'$ iff the required $\rho'$ exists. So both sums are non-empty under the same conditions and have at most one summand (since $\rho'$ is unique), and if it is the case these summands are equal by \lemmaword~\ref{lem:symbolic_unification_soundness}, the inductive hypothesis and the fact that the value $D$ is stable w.r.t. renaming of variables (it is a generalization of \lemmaword~\ref{lem:measures_changing_env} that follows simply from \lemmaword~\ref{lem:gen_measures_changing_env}):
	
\[ \begin{array}{l}
D\,(\taskst{g'}{\mkenv{\sigma \rho \delta'}{n}}, \Gamma') \stackrel{\text{\lemmaword~\ref{lem:symbolic_unification_soundness}}}{=} D\,(\taskst{g'}{\mkenv{\sigma \delta \rho' \tau}{n}}, \Gamma') = \\
= D\,(\taskst{g'}{\mkenv{\sigma \delta \rho'}{n}}, \Gamma') \stackrel{\text{ind.hyp.}}{=} \mathcal{D}(\schemewithvset{\mathfrak{S}}{U})(\rho')
\end{array} \] 

	\end{enumerate}
	
	\item Let $\Gamma = g' : \Gamma'$ and $g = (\invokegoal{R^k}{t_1}{t_k})$.
	
	In this case we have \[ d\,(\taskst{\invokegoal{R^k}{t_1}{t_k}}{\mkenv{\sigma\rho}{n}}) + \smashoperator{\sum\limits_{a \in \tra{\taskst{\invokegoal{R^k}{t_1}{t_k}}{\mkenv{\sigma\rho}{n}}}}} D\,(\taskst{g'}{a}, \Gamma') = \]
	
	\[ = d\,(\taskst{\invokegoal{R^k}{t_1 \sigma \rho}{t_k \sigma \rho}}{e_{init}}) +
      \smashoperator{\sum\limits_{\substack{ \rho' \colon U \to \grterms \\
                                      \rho' \succ \rho \\
                                      (t_1 \sigma \rho', \dots, t_k \sigma \rho') \in \sembr{R^k}  }}}
           \mathcal{D}(\schemewithvset{\mathfrak{S}}{U})(\rho')  \]
           
    where
    
    \[ \begin{array}{lcl}
    U & = &  V \cup \bigcup_{i} FV\,(t_i \sigma) \\
    \end{array} \]
    
    The left summands are equal by \lemmaword~\ref{lem:measures_changing_env}.
    
    Each $a \in \tra{\taskst{\invokegoal{R^k}{t_1}{t_k}}{\mkenv{\sigma\rho}{n}}}$ has the form $(\sigma \rho \delta, n')$ for some $\delta$ and some $n' > n$ by \lemmaword~\ref{lem:gen_measures_changing_env}. All free variables of terms $t_i$ are mapped into ground terms in each delta, because the relational call is grounding.
    
    There is a bijection between the set of answers under the sum on the lhs and a set of suitable $\rho'$ at the rhs: all ground values for the free variables that are consistent with the denotational semantics are present in some answer (because of completeness of the operational semantics w.r.t. to the denotational one) and in exactly one (because the call is answer-unique).
    
    Now, we need to show that if we take any answer $(\sigma \rho \delta, n')$ and $\rho'$ that corresponds to it by the described bijection then the values for them under the sums are equal. First, we need to replace $D\,(\taskst{g'}{(\sigma \rho \delta, n')}, \Gamma')$ by $D\,(\taskst{g'}{(\sigma \rho \delta\restriction_U, n)}, \Gamma')$, we can do it because all variables from $\Dom\,(\delta) \setminus U$ have indices greater or equal than $n$ (by \lemmaword~\ref{lem:update_substitutions_FV}) and therefore these variables are not free variables of goals from $(g' : \Gamma')$, so the values of $d$ on this goals will be the same for the restricted substitution (it follows simply from \lemmaword~\ref{lem:gen_measures_changing_env}). After this replacement we can directly apply the inductive hypothesis.
    \qed
	
	\end{enumerate}

\end{enumerate}
\end{proof}

\end{document}